\documentclass[12pt,journal,final,onecolumn]{IEEEtran}
\usepackage[tight,footnotesize,center]{subfigure}
\usepackage{amsmath,amsfonts,amssymb,bbm}
\usepackage[footnotesize]{caption}
\usepackage{graphicx,xcolor}
\usepackage{multirow,bigdelim,longtable}
\usepackage{amsthm,bm}
\usepackage{hhline}
\usepackage[nobreak]{cite}
\usepackage{tikz,pgfplots}
\usepackage{flushend}
\usetikzlibrary{intersections,calc,patterns,arrows}
\theoremstyle{definition}
\newtheorem{thm}{Theorem}
\newtheorem{lem}{Lemma}

\newtheorem*{thm*}{Theorem}
\newtheorem{remark}{Remark}
\usepackage{float}
\restylefloat{table}
\newcommand{\mc}[1]{\mathcal{#1}}

\newcommand{\msf}[1]{\mathsf{#1}}
\newcommand{\squeezeup}{\vspace{0mm}}

\begin{document}
\title{Multicoding Schemes for Interference Channels}
\author{Ritesh Kolte, Ayfer \"{O}zg\"{u}r, Haim Permuter%
\thanks{R.~Kolte and A.~\"{O}zg\"{u}r are with the Department of Electrical Engineering at Stanford University. H.~Permuter is with the Department of Electrical and Computer Engineering at Ben-Gurion University of the Negev. This work was presented in part in CISS 2014 Princeton NJ\cite{Kol14b} and ISIT 2014 Honolulu HI \cite{Kol14}.
The work of R. Kolte and A. \"{O}zg\"{u}r was supported in part by a Stanford Graduate fellowship and NSF CAREER award \#1254786. The work of H. Permuter was supported by the Israel Science Foundation (grant no. 684/11) and the ERC starting grant.}
}
\maketitle               

\begin{abstract}
The best known inner bound for the 2-user discrete memoryless interference channel is the Han-Kobayashi rate region. The coding schemes that achieve this region are based on rate-splitting and superposition coding. In this paper, we develop a multicoding scheme to achieve the same rate region. A key advantage of the multicoding nature of the proposed coding scheme is that it can be naturally extended to more general settings, such as when encoders have state information or can overhear each other. In particular, we extend our coding scheme to characterize the capacity region of the state-dependent deterministic Z-interference channel when noncausal state information is available at the interfering transmitter. We specialize our results to the case of the linear deterministic model with on/off interference which models a wireless system where a cognitive transmitter is noncausally aware of the times it interferes with a primary transmission. For this special case, we provide an explicit expression for the capacity region and discuss some interesting properties of the optimal strategy. We also extend our multicoding scheme to find the capacity region of the deterministic Z-interference channel when the signal of the interfering transmitter can be overheard at the other transmitter (a.k.a. unidirectional partial cribbing).
\end{abstract}


\section{Introduction}
The discrete memoryless interference channel (DM-IC) is the canonical model for studying the effect of interference in wireless systems. The capacity of this channel is only known in some special cases e.g. class of deterministic ICs \cite{Gam82,Cho07}, strong interference conditions \cite{Sat81,Cos87,Chu07}, degraded conditions \cite{Ben79,Liu08} and a class of semideterministic ICs \cite{Cho09}. Characterizing the capacity region in the general case has been one of the long standing open problems in information theory. The best known achievable rate region is the so-called Han-Kobayashi scheme, which can be achieved by using schemes that are based on the concepts of rate-splitting and superposition coding \cite{Han81,Cho06}. Rate-splitting refers to the technique of splitting the message at a transmitter into a common and a private part, where the common part is decoded at all the receivers and the private part is decoded only at the intended receiver. The two parts of the message are then combined into a single signal using superposition coding, first introduced in \cite{Cov72} in the context of the broadcast channel. 
In all the special cases where the capacity is known, the Han-Kobayashi region equals the capacity region. However, it has been very recently shown that this inner bound is not tight in general \cite{Nai15}.

The first result we present in this paper is to show that  the Han-Kobayashi region can be achieved by a multicoding scheme. This scheme does not involve any explicit rate-splitting. Instead, the codebook at each encoder is generated as a multicodebook, i.e. there are multiple codewords corresponding to each message. The auxiliary random variable in this scheme does not explicitly carry a part of the message, rather it implicitly carries \emph{some} part of the message, and it is not required to specify which part.\footnote{A similar idea, combined with block-Markov operation, has been recently used in \cite{Lim14} to develop an achievability scheme called distributed-decode-forward for broadcast traffic on relay networks.} In this sense, it's role is different from that in the Han-Kobayashi scheme \cite{Han81,Cho06}, and is reminiscent of the encoding for state-dependent channels in \cite{Gel80}, and the alternative proof of Marton's achievable rate region for the broadcast channel given in \cite{Gam81}. A key advantage of
the multicoding nature of the new scheme is that it can be easily extended to obtain simple achievability schemes for setups in which the canonical interference channel model is augmented to incorporate additional node capabilities such as cognition and state-dependence, while extending the original Han-Kobayashi scheme to such setups can quickly become highly involved. We demonstrate this by constructing schemes for settings which augment the canonical interference channel model in different ways.
 

The first setting we consider is when the interference channel is state-dependent and the state-information is available non-causally to one of the transmitters (cognitive transmitter). For simplicity, we focus on the case when the cross-link between the non-cognitive transmitter and its undesired receiver is weak enough to be ignored, giving rise to the so called $Z$-interference channel topology. We know that for a point-to-point state-dependent channel with non-causal state information at the encoder, the optimal achievability scheme due to Gelfand and Pinsker uses multicoding at the encoders. Hence, for state-dependent interference channels with noncausal state information at the encoders too, we would like to use the idea of multicoding. Since the new achievability scheme that we present for the canonical interference channel already involves multicoding, it requires almost no change to be applicable to the state-dependent setting. Apart from being simple, we are also able to prove its optimality for the case of the deterministic $Z$-interference channel.

We then specialize our capacity characterization for the state-dependent deterministic $Z$-interference channel to the case where the channels are governed by the linear deterministic model of \cite{Ave11}. In the recent literature, this model has proven extremely useful for approximating the capacity of wireless networks and developing insights for the design of optimal communication strategies. We consider a linear deterministic Z-interference channel, in which the state of the channel denotes whether the interference link is present or not. When the transmitters are base-stations and the receivers are end-users, this can model the scenario where one of the transmitters is cognitive, for example it can be a central controller that knows when the other Tx-Rx pair will be scheduled to communicate on the same frequency band. When the two Tx-Rx pairs are scheduled to communicate on the same frequency band, this gives an interference channel; when they communicate on different frequency bands each pair gets a clean channel free of interference. Moreover, the cognitive transmitter can know the schedule ahead of time, i.e. the times at which its transmission will be interfering with the second Tx-Rx pair. For this special case, we identify auxiliary random variables and provide an explicit expression for the capacity region. This explicit capacity characterization allows us to identify interesting properties of the optimal strategy. In particular, with single bit level for the linear deterministic channels (which would imply low to moderate SNR for the corresponding Gaussian channels), the sum rate is maximized when the interfering transmitter remains silent (transmits $0$'s) at times when it interferes with the second transmission. It then treats these symbols as stuck to $0$ and performs Gelfand-Pinsker coding. The second transmitter observes a clean channel at all times and communicates at the maximal rate of $1$ bit per channel use. This capacity characterization also reveals that when all nodes are provided with the state information the sum-capacity cannot be further improved. Thus, for this
channel, the sum-capacity when all nodes have state information is the same as that when only the interfering encoder
has state information. 

Motivated by wireless applications, there has been significant recent interest in state-dependent interference channels (ICs), where the state information is known only to some of the transmitters. Given the inherent difficulty of the problem, many special cases have been considered  \cite{Zha13,Goo13,Dua13a,Dua13b}, for which different coding schemes have been proposed. However, exact capacity characterizations have proven difficult. Another line of related work has been the study of cognitive state-dependent ICs  \cite{Rin11,Som08,Dua12,Kaz13}. Here, the term ``cognitive'' is usually used to mean that the cognitive transmitters know not only the state of the channel but also messages of other transmitters. Note that this assumption is significantly stronger than assuming state information at the transmitter as we do here.

The second setting we consider is when one of the transmitters has the capability to overhear the signal transmitted by the other transmitter, which can be used to induce cooperation between the two transmitters. This is different from having orthogonal communication links (or conferencing) between the encoders, as studied in \cite{Wan11b}. Instead, overhearing exploits the natural broadcasting nature of the wireless medium to establish cooperation without requiring any dedicated resources. A variety of different models have been used to capture overhearing \cite{Pra11,Yan11,Car12}, and are known by different names such as cribbing, source cooperation, generalized feedback, cognition etc. We use "partial cribbing" to model the overhearing, in which some deterministic function of the signal transmitted by the non-cognitive transmitter is available at the cognitive transmitter in a strictly causal fashion. Again, for simplicity, we focus on the case of the $Z$-interference channel, where the cross-link between the non-cognitive transmitter and its undesired receiver is weak enough to be ignored. For this setting, we develop a simple achievability scheme by combining our multicoding-based scheme with block-Markov coding and show that it is optimal for deterministic configurations.

Finally, to further illustrate the point that simple schemes can be obtained for augmented scenarios, we describe two extensions which introduce even more complexity in the model. In the first extension, a third message is introduced in  the state-dependent Z-interference channel, which is to be communicated from the interfering transmitter to the interfered receiver. The second extension combines the state-dependent Z-IC and the Z-IC with unidirectional partial cribbing. In both extensions, we are able to obtain simple optimal schemes by naturally extending the multicoding-based achievability schemes.

\subsection*{Organization}
We describe the models considered in this paper formally in Section~\ref{sec:model}. 
The alternate achievability scheme that achieves the Han-Kobayashi region is presented in Sections~\ref{sec:outline}.  Section~\ref{sec:state} describes the results concerning the state-dependent setup and section~\ref{sec:cribbing} describes the results concerning the cribbing setup. The two extensions are described in Section~\ref{sec:extensions} and we end the paper with a short discussion in Section~\ref{sec:conclude}.

\section{Model}\label{sec:model}

Capital letters, small letters and capital calligraphic letters denote random variables, realizations and alphabets respectively. The tuple $(x(1),x(2),\dots ,x(n))$ and the set $\{a,{a+1},\dots ,b\}$ are denoted by $x^n$ and $[a:b]$ respectively, and $\mc{T}_{\epsilon}^{(n)}$ stands for the $\epsilon$-strongly typical set of length-$n$ sequences.

We now describe the channel models considered in this paper.

\subsection{Canonical Interference Channel}
The two-user discrete memoryless interference channel $p_{Y_1,Y_2|X_1,X_2}(y_1,y_2|x_1,x_2)$ is depicted in Fig.~\ref{fig:model}. Each sender $j\in\{1,2\}$ wishes to communicate a message $M_j$ to the corresponding receiver. 

A $(n,2^{nR_1},2^{nR_2},\epsilon)$ code for the above channel consists of the encoding and decoding functions:
\begin{IEEEeqnarray*}{rCl}
f_{j,i} & : & [1:2^{nR_j}] \rightarrow \mc{X}_j, \quad j\in\{1,2\}, 1\leq i\leq n,\\ 
g_j & : & \mc{Y}_j^n \rightarrow [1:2^{nR_j}],\quad j\in\{1,2\},
\end{IEEEeqnarray*}
such that 
$$\text{Pr}\left\{g(Y_j^n)\neq M_j\right\} \leq \epsilon,\quad j\in\{1,2\},$$
where $M_1$ and $M_2$ are assumed to be distributed uniformly in $[1:2^{nR_1}]$ and $[1:2^{nR_2}]$ respectively. A rate pair $(R_1,R_2)$ is said to be \emph{achievable} if for every $\epsilon > 0,$ there exists a $(n,2^{nR_1},2^{nR_2},\epsilon)$ code for sufficiently large $n$. The capacity region is defined to be the closure of the achievable rate region.

\begin{figure}[!ht]
\centering
\includegraphics[scale=1.5]{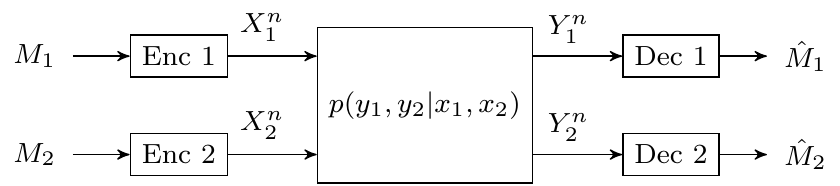}
\caption{Two-User Discrete Memoryless Interference Channel (DM-IC)}
\label{fig:model}
\end{figure}


\subsection{State-Dependent Z-Interference Channel}\label{subsec:model_state}

The discrete memoryless Z-interference channel $p(y_1|x_1,s)p(y_2|x_1,x_2,s)$ with discrete memoryless state $p(s)$ is depicted in Fig.~\ref{fig:model_gen}. The states are assumed to be known noncausally at encoder 1. Each sender $j\in\{1,2\}$ wishes to communicate a message $M_j$ at rate $R_j$ to the corresponding receiver. For this setting, a $(n,2^{nR_1},2^{nR_2},\epsilon)$ code consists of the encoding and decoding functions:
\begin{IEEEeqnarray*}{rCl}
f_{1,i}& : &[1:2^{nR_1}]\times \mc{S}^n \rightarrow \mc{X}_1, \quad  1\leq i\leq n,\\ 
f_{2,i}& :& [1:2^{nR_2}] \rightarrow \mc{X}_2, \quad 1\leq i\leq n,\\ 
g_j &: & \mc{Y}_j^n \rightarrow [1:2^{nR_j}],\quad j\in\{1,2\},
\end{IEEEeqnarray*}
such that 
$$\text{Pr}\left\{g(Y_j^n)\neq M_j\right\} \leq \epsilon,\quad j\in\{1,2\}.$$ The probability of error, achievable rate pairs $(R_1,R_2)$ and the capacity region are defined in a similar manner as before.

\begin{figure}[!h]
\centering
\includegraphics[scale=1.5]{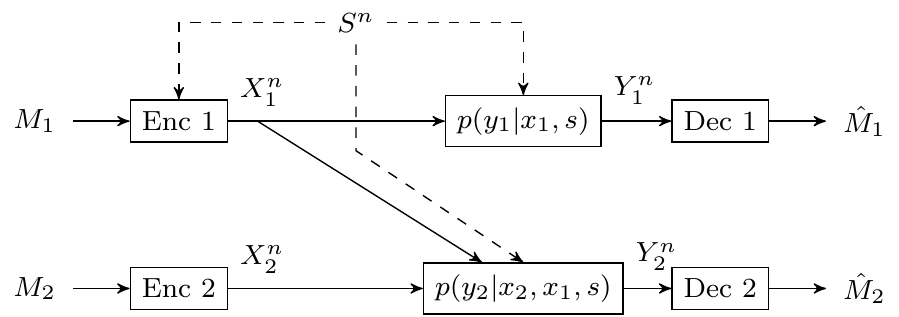}
\caption{The State-Dependent Z-Interference Channel (S-D Z-IC)}
\label{fig:model_gen}\squeezeup
\end{figure}

The deterministic S-D Z-IC is depicted in Fig.~\ref{fig:model_state_det}. The channel output $Y_1$ is a deterministic function $y_1(X_1,S)$ of the channel input $X_1$ and the state $S$. At receiver 2, the channel output $Y_2$ is a deterministic function $y_2(X_2,T_1)$ of the channel input $X_2$ and the interference $T_1$, which is assumed to be a deterministic function $t_1(X_1,S)$. 
We also assume that if $x_2$ is given, $y_2(x_2,t_1)$ is an injective function of $t_1$, i.e. there exists some function $g$ such that $t_1=g(y_2,x_2).$ 

\begin{figure}[!h]
\centering
\includegraphics[scale=1.5]{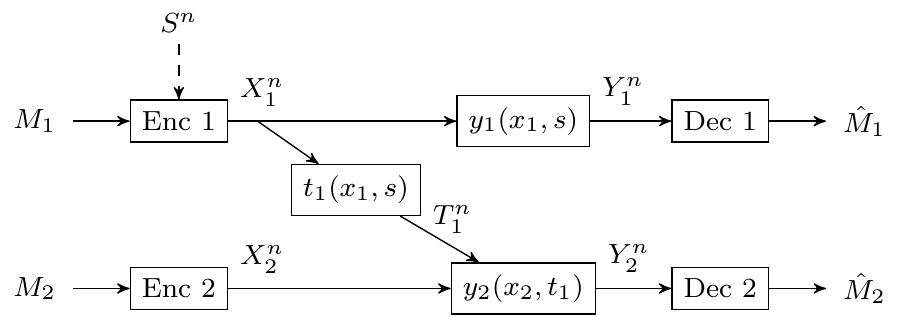}
\caption{The Injective Deterministic S-D Z-IC}
\label{fig:model_state_det}
\end{figure}

We consider a special case of the injective deterministic S-D Z-IC in detail, which is the modulo-additive S-D Z-IC, depicted in Fig.~\ref{fig:model_modulo}. All channel inputs and outputs come from a finite alphabet $\mc{X}=\{0,1,\dots ,|\mc{X}|-1\}$. The channel has two states. In state $S=0$, there is no interference while in state $S=1$, the cross-link is present. When the cross-link is present, the output at receiver~2 is the modulo-$\mc{X}$ sum of $X_2$ and $X_1$. For all other cases, the output is equal to the input. We can describe this formally as: 
\begin{equation*}
\begin{split}
Y_1 & = X_1,\\
Y_2 & = X_2 \oplus (S\cdot X_1).
\end{split}
\end{equation*}
Assume that the state $S$ is i.i.d. Ber$(\lambda)$. A generalization of this model that incorporates multiple levels is also considered subsequently.

\begin{figure}[!h]
\centering
\includegraphics[scale=1.5]{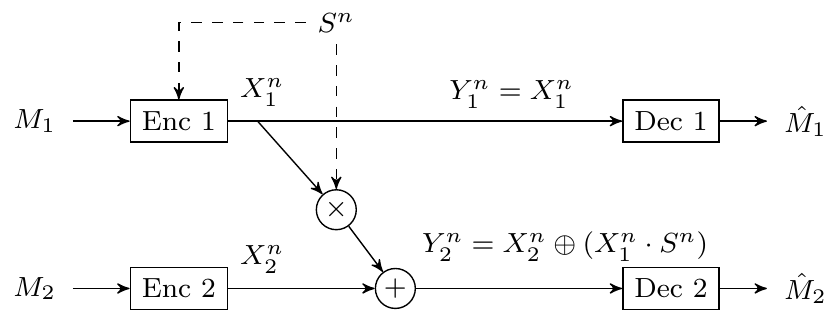}
\caption{The Modulo-Additive S-D Z-IC. All channel inputs and outputs take values in the same finite alphabet $\mc{X}$. The state $S$ is Ber$(\lambda).$}
\label{fig:model_modulo}
\end{figure}

\subsection{Z-Interference Channel with Partial Cribbing}\label{subsec:model_crib}
The discrete memoryless deterministic Z-interference channel is depicted in Fig.~\ref{fig:model_crib}. The channel output $Y_1$ is a deterministic function $y_1(X_1)$ of the channel input $X_1$. At receiver 2, the channel output $Y_2$ is a deterministic function $y_2(X_2,T_1)$ of the channel input $X_2$ and the interference $T_1$, which is assumed to be a deterministic function $t_1(X_1)$. We also assume that if $x_2$ is given, $y_2(x_2,t_1)$ is an injective function of $t_1$, i.e. there exists some function $g$ such that $t_1=g(y_2,x_2).$ Each sender $j\in\{1,2\}$ wishes to communicate a message $M_j$ at rate $R_j$ to the corresponding receiver. 

We assume that encoder 1 can overhear the signal from transmitter 2 \emph{strictly causally}, which is modeled as partial cribbing with a delay \cite{Asn13}. The partial cribbing signal, which is a function of $X_2$ is denoted by $Z_2$. So $X_{1i}$ is a function of $(M_1,Z_2^{i-1})$ and $X_{2i}$ is a function of $M_2$.

\begin{figure}[!ht]
\centering
\includegraphics[scale=1.5]{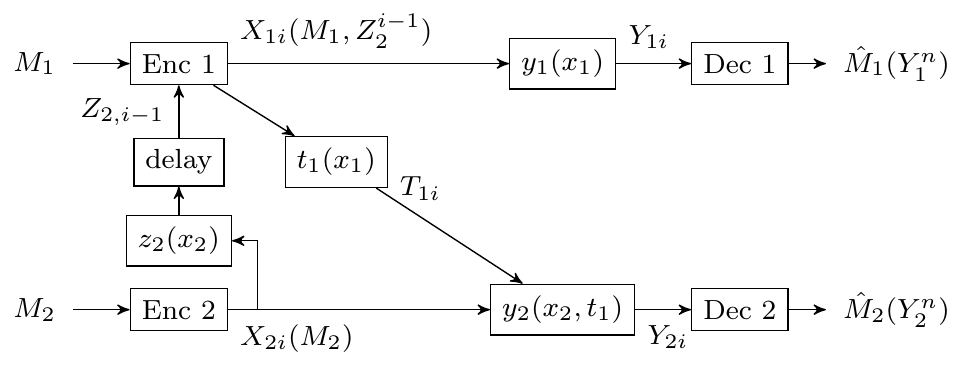}
\caption{Injective Deterministic Z-Interference Channel with Unidirectional Partial Cribbing}
\label{fig:model_crib}
\end{figure}

A $(n,2^{nR_1},2^{nR_2},\epsilon)$ code for this setting consists of 
\begin{IEEEeqnarray*}{rCl}
f_{1,i} & : & [1:2^{nR_1}]\times \mc{Z}_2^{i-1} \rightarrow \mc{X}_1, \quad  1\leq i\leq n,\\ 
f_{2,i} & : & [1:2^{nR_2}] \rightarrow \mc{X}_2, \quad 1\leq i\leq n,\\ 
g_j & : &  \mc{Y}_j^n \rightarrow [1:2^{nR_j}],\quad j\in\{1,2\},
\end{IEEEeqnarray*}
such that 
$$\text{Pr}\left\{g(Y_j^n)\neq M_j\right\} \leq \epsilon,\quad j\in\{1,2\}.$$ The probability of error, achievable rate pairs $(R_1,R_2)$ and the capacity region are defined in a similar manner as before.

\section{Canonical Interference Channel}\label{sec:outline}

\subsection{Preliminaries}\label{sec:prelim}
The currently best known achievable rate region for the 2-user DM-IC was provided by Han and Kobayashi in \cite{Han81}, using a scheme based on rate-splitting and superposition coding. An alternative achievable rate region that included the Han-Kobayashi rate region was proposed in \cite{Cho06}, using another scheme that used rate-splitting and superposition coding. Using the terminology introduced in \cite{Wan13}, the encoding in \cite{Han81} can be described as employing \emph{homogeneous} superposition coding, while that in \cite{Cho06} can be described as employing \emph{heterogeneous} superposition coding. It was then proved in \cite{Cho08} that the two regions are, in fact, equivalent and given by the following compact representation (see also \cite{Kra06,Kob07}).

\begin{thm}[Han-Kobayashi Region]\label{thm:HK}
A rate pair $(R_1,R_2)$ is achievable for the DM-IC $p(y_1,y_2|x_1,x_2)$ if
\begin{equation}\label{eq:achreg_prelim}
\begin{split}
R_1 & < I(X_1;Y_1|U_2,Q),\\
R_2 & < I(X_2;Y_2|U_1,Q),\\
R_1 + R_2 & < I(X_1;Y_1|U_1,U_2,Q) +I(X_2,U_1;Y_2|Q) ,\\
R_1 + R_2 & < I(X_1,U_2;Y_1|U_1,Q) + I(X_2,U_1;Y_2|U_2,Q),\\
R_1 + R_2 & < I(X_1,U_2;Y_1|Q) + I(X_2;Y_2|U_1,U_2,Q),\\
2R_1 + R_2 & < I(X_1;Y_1|U_1,U_2,Q) + I(X_2,U_1;Y_2|U_2,Q) + I(X_1,U_2;Y_1|Q),\\
R_1 + 2R_2 & < I(X_2;Y_2|U_1,U_2,Q) + I(X_1,U_2;Y_1|U_1,Q) + I(X_2,U_1;Y_2|Q),
\end{split}
\end{equation}
for some pmf $p(q)p(u_1,x_1|q)p(u_2,x_2|q),$ where ${|\mc{U}_1|\leq |\mc{X}_1|+4}$, ${|\mc{U}_2|\leq |\mc{X}_2|+4}$ and ${|\mc{Q}|\leq 4.}$
\end{thm}

\subsection{Outline of the new achievability scheme}
We first describe the alternative achievability scheme informally and discuss the similarities and differences with the existing achievability schemes. The later subsections describe and analyze the scheme formally.

Encoder $j$, where $j\in\{1,2\}$ prepares two codebooks: 
\begin{itemize}
\item A transmission multicodebook\footnote{The term ``multicodebook'' refers to the fact that there are multiple codewords corresponding to each message.}, which is a set of codewords $\{x_j^n(\cdot,\cdot)\}$ formed using the transmission random variable $X_j$. This set is partitioned into a number of bins (or subcodebooks), where the bin-index corresponds to the message, 
\item A coordination codebook which is a set of codewords $\{u_j^n(\cdot)\}$ formed using the auxiliary random variable $U_j$.
\end{itemize}
Given a message, one codeword $x_j^n$ from the corresponding bin in the transmission multicodebook is chosen so that it is jointly typical with some sequence $u_j^n$ in the coordination codebook. The codeword $x_j^n$ so chosen forms the transmission sequence.

At a decoder, the desired message is decoded by using joint typicality decoding, which uses the coordination codebook and the transmission multicodebook of the corresponding encoder and the coordination codebook of the other encoder. Thus, a receiver makes use of the interference via its knowledge of the coordination codebook at the interfering transmitter.

From the above description, it can be seen that the coordination codebook does not carry any message. Its purpose is to ensure that the transmission sequence from a given bin is well-chosen, i.e. it is beneficial to the intended receiver and also the unintended receiver. To the best of our knowledge, this is the first time an auxiliary random variable (which is not the time-sharing random variable) appears in one of the best known achievability schemes without being explicitly associated with any message.

\subsection{Achievability scheme}\label{subsec:achHK}
Choose a pmf $p(u_1,x_1)p(u_2,x_2)$ and $0<\epsilon'<\epsilon$.
\subsubsection*{Codebook Generation}
\begin{itemize}
\item Encoder 1 generates a coordination codebook consisting of $2^{nR_{1c}}$ codewords\footnote{Though there is no notion of a common message or a private message in this achievability scheme, we use the subscripts $c$ and $p$ to convey if the corresponding random variables are used for decoding at all destinations or only the desired destination respectively.} $u_1^n(l_{1c}),\;l_{1c}\in[1:2^{nR_{1c}}]$ i.i.d. according to $\prod_{i=1}^np(u_{1i})$. It also generates a transmission multicodebook consisting of $2^{n(R_1+R_{1p})}$ codewords $x_1^n(m_1,l_{1p}),\;m_1\in[1:2^{nR_1}],\; l_{1p}\in[1:2^{nR_{1p}}]$ i.i.d. according to $\prod_{i=1}^np(x_{1i})$.
\item Similarly, encoder 2 generates a coordination codebook consisting of $2^{nR_{2c}}$ codewords $u_2^n(l_{2c}),\;l_{2c}\in[1:2^{nR_{2c}}]$ i.i.d. according to $\prod_{i=1}^np(u_{2i})$. It also generates a transmission multicodebook consisting of $2^{n(R_2+R_{2p})}$ codewords $x_2^n(m_2,l_{2p}),\;m_2\in[1:2^{nR_2}],\; l_{2p}\in[1:2^{nR_{2p}}]$ i.i.d. according to $\prod_{i=1}^np(x_{2i})$.
\end{itemize}

\subsubsection*{Encoding}
\begin{itemize}
\item To transmit message $m_1$, encoder 1 finds a pair $(l_{1c},l_{1p})$ such that $$(u_1^n(l_{1c}),x_1^n(m_1,l_{1p}))\in\mc{T}^{(n)}_{\epsilon'}$$ and transmits $x_1^n(m_1,l_{1p})$. If it cannot find such a pair, it transmits $x_1^n(m_1,1)$.
\item Similarly, to transmit message $m_2$, encoder 2 finds a pair $(l_{2c},l_{2p})$ such that $$(u_2^n(l_{2c}),x_2^n(m_2,l_{2p}))\in\mc{T}^{(n)}_{\epsilon'}$$ and transmits $x_2^n(m_2,l_{2p})$. If it cannot find such a pair, it transmits $x_2^n(m_2,1)$.
\end{itemize}

The codebook generation and encoding process are illustrated in Fig.~\ref{fig:encoding}.

\begin{figure}[t]
\centering
\includegraphics[scale=1.5]{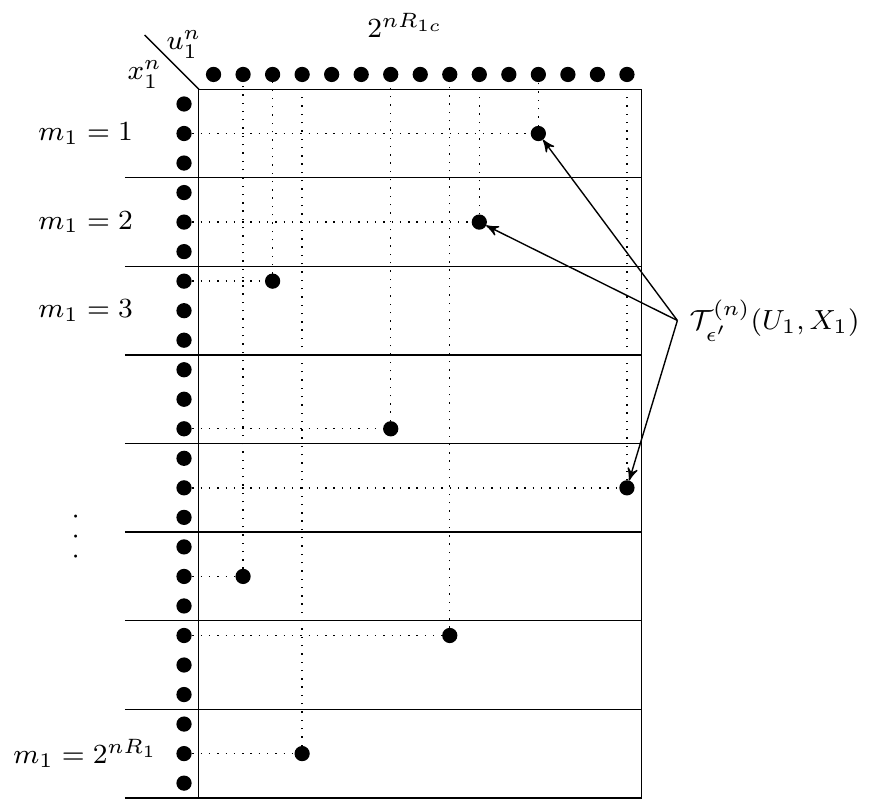}
\vspace{1mm}
\caption{Codebook Generation and Encoding at Encoder~1. The independently generated $x_1^n$ sequences, lined up vertically in the figure, are binned into $2^{nR_1}$ bins. The independently generated coordination sequences $u_1^n$ are lined up horizontally. To transmit message $m_1$, a jointly typical pair $(x_1^n,u_1^n)$ is sought where $x_1^n$ falls into the $m_1$-th bin, and then $x_1^n$ is transmitted.}
\label{fig:encoding}
\end{figure}

\subsubsection*{Decoding}
\begin{itemize}
\item Decoder 1 finds the unique $\hat{m}_1$ such that $$(u_1^n(l_{1c}),x_1^n(\hat{m}_1,l_{1p}),u_2^n(l_{2c}),y_1^n)\in\mc{T}^{(n)}_{\epsilon}$$ for some $(l_{1c},l_{1p},l_{2c})$. If none or more than one such $\hat{m}_1$ are found, then decoder 1 declares error.
\item Decoder 2 finds the unique $\hat{m}_2$ such that $$(u_2^n(l_{2c}),x_2^n(\hat{m}_2,l_{2p}),u_1^n(l_{1c}),y_2^n)\in\mc{T}^{(n)}_{\epsilon}$$ for some $(l_{2c},l_{2p},l_{1c})$. If none or more than one such $\hat{m}_2$ are found, then decoder 2 declares error.
\end{itemize}

\subsubsection*{Discussion}
Before providing the formal analysis of the probability of error to show that the coding scheme described above achieves the Han-Kobayashi region, we discuss the connection between the new scheme and the scheme from \cite{Cho06} which motivates the equivalence of their rate regions.

Consider the set of codewords used at encoder 1. While this set resembles a multicodebook, it can be reduced to a standard codebook (one codeword per message) by stripping away the codewords in each bin that are not jointly typical with any of the $u_1^n$ sequences, and therefore are never used by the transmitters. In other words, after we generate the multicodebook in Fig.~\ref{fig:encoding}, we can form a smaller codebook by only keeping one codeword per message which is jointly typical with one of the $u_1^n$ sequences (i.e., those codewords highlighted in Fig.~\ref{fig:encoding}). Note that this reduced codebook indeed has a superposition structure. Each of the $ 2^{nR_1} $ remaining codewords  $x_1^n$ is jointly typical with one of the $2^{nR_{1c}}$  $u_1^n$ codewords, and when $n$ is large there will be exactly $ 2^{n(R_1-R_{1c})} $   $x_1^n$ sequences that are typical with each $u_1^n$ sequence, i.e., these $ 2^{n(R_1-R_{1c})} $   $x_1^n$ sequences will look as if they were generated i.i.d. from $p(x_1|u_1)$. Therefore, the $u_1^n$ sequences can be indeed thought as the cloud centers in this superposition codebook and $x_1^n$'s as the satellite codewords. Therefore, our multicodebook construction can be viewed as an equivalent way to generate a superposition codebook as in \cite{Cho08}. This reveals that both the codebook structure and the decoding in our scheme are similar to that in the Han-Kobayashi scheme and therefore the two achievable rate regions are, not surprisingly, equal.

However, note that for broadcast channels, combining Marton coding (which employs multicoding) \cite{Gam81} with Gelfand-Pinsker coding (which also employs multicoding) is more straightforward than combining superposition coding with Gelfand-Pinsker coding. The former has been shown to be optimal in some cases \cite{Lap13}. Since our codebook construction for the interference channel also has the flavor of multicoding, extending this construction to setups where multicoding is required is also quite straightforward. As mentioned in the introduction, we exploit this to develop simple achievability schemes for more general setups described in later sections.

\subsubsection*{Probability of Error}
Due to the symmetry of the code, the average probability of error $\msf{P}(\mc{E})$ is equal to $\msf{P}(\mc{E}|M_1,M_2)$, so we can assume $(M_1,M_2) = (1,1)$ and analyze $\msf{P}(\mc{E}|1,1)$. Let $(L_{1c},L_{1p},L_{2c},L_{2p})$ denote the indices chosen during encoding by encoder 1 and encoder 2. 

We now define events that cover the event of error in decoding message $m_1$:
\begin{IEEEeqnarray*}{rCl}
\mc{E}_1 & \triangleq & \{(U_1^n(l_{1c}),X_1^n(1,l_{1p}))\notin\mc{T}^{(n)}_{\epsilon'} \;\text{ for all } l_{1c}, l_{1p}\}, \\
\mc{E}_2 & \triangleq & \{(U_1^n(L_{1c}),X_1^n(1,L_{1p}),U_2^n(L_{2c}),Y_1^n)\notin\mc{T}^{(n)}_{\epsilon}\}, \\
\mc{E}_3 & \triangleq & \{(U_1^n(L_{1c}),X_1^n(m_1,l_{1p}),U_2^n(L_{2c}),Y_1^n)\in\mc{T}^{(n)}_{\epsilon}\text{ for some }m_1\neq 1, \text{ for some } l_{1p}\}, \\
\mc{E}_4 & \triangleq & \{(U_1^n(L_{1c}),X_1^n(m_1,l_{1p}),U_2^n(l_{2c}),Y_1^n)\in\mc{T}^{(n)}_{\epsilon} \text{ for some }m_1\neq 1, \text{ for some } l_{1p},l_{2c}\} ,\\
\mc{E}_5 & \triangleq & \{(U_1^n(l_{1c}),X_1^n(m_1,l_{1p}),U_2^n(L_{2c}),Y_1^n)\in\mc{T}^{(n)}_{\epsilon} \text{ for some }m_1\neq 1, \text{ for some } l_{1p},l_{1c}\} ,\\
\mc{E}_6 & \triangleq & \{(U_1^n(l_{1c}),X_1^n(m_1,l_{1p}),U_2^n(l_{2c}),Y_1^n)\in\mc{T}^{(n)}_{\epsilon} \text{ for some }m_1\neq 1, \text{ for some } l_{1c},l_{1p},l_{2c}\}.
\end{IEEEeqnarray*}

Consider also the event $\mc{E}'_1$, analogous to $\mc{E}_1$, which is defined as follows.
\begin{IEEEeqnarray*}{rCl}
\mc{E}'_1 & \triangleq & \{(U_2^n(l_{2c}),X_2^n(1,l_{2p}))\notin\mc{T}^{(n)}_{\epsilon'} \;\text{ for all } l_{2c}, l_{2p}\}. \label{eq:E'1}
\end{IEEEeqnarray*}

Since an error for $m_1$ occurs only if at least one of the above events occur, we use the union bound to get the following upper bound on the average probability of error in decoding $m_1$:
$$ \msf{P}(\mc{E}_1) + \msf{P}(\mc{E}'_1)+ \msf{P}(\mc{E}_2\cap\mc{E}_1^c\cap\mc{E}'^{c}_1) + \msf{P}(\mc{E}_3) + \msf{P}(\mc{E}_4) + \msf{P}(\mc{E}_5) + \msf{P}(\mc{E}_6).$$

By the mutual covering lemma \cite[Chap. 8]{Gam12}, $\msf{P}(\mc{E}_1)\rightarrow 0$ as $n\rightarrow\infty$ if 
\begin{IEEEeqnarray}{rCl}
R_{1p} + R_{1c} & > & I(U_1;X_1) + \delta(\epsilon'),\label{eq:ach1}
\end{IEEEeqnarray}
where $\delta(\epsilon')\rightarrow 0$ as $\epsilon'\rightarrow 0.$

Similarly, we get that 
$\msf{P}(\mc{E}'_1)\rightarrow 0$ as $n\rightarrow\infty$ if 
\begin{IEEEeqnarray}{rCl}
R_{2p} + R_{2c} & > & I(U_2;X_2) + \delta(\epsilon').\label{eq:ach1'}
\end{IEEEeqnarray}

By the conditional typicality lemma, $\msf{P}(\mc{E}_2\cap\mc{E}_1^c\cap\mc{E}'^{c}_1)$ tends to zero as $n\rightarrow\infty$.

For $\msf{P}(\mc{E}_3)\rightarrow 0$, we can use the packing lemma from \cite[Ch. 3]{Gam12} to get the condition
\begin{equation}\label{eq:ach2}
R_1 + R_{1p}  <  I(X_1;U_1,U_2,Y_1) - \delta(\epsilon),
\end{equation}where $\delta(\epsilon)\rightarrow 0$ as $\epsilon\rightarrow 0.$

For $\msf{P}(\mc{E}_4)\rightarrow 0$, we can again use the packing lemma to get the condition
\begin{equation}\label{eq:ach3}
R_1 + R_{1p} + R_{2c} < I(X_1,U_2;U_1,Y_1)- \delta(\epsilon).
\end{equation}

For $\msf{P}(\mc{E}_5)\rightarrow 0$, we apply the multivariate packing lemma from the Appendix as shown in \eqref{eq:multipack_2} to get the condition
\begin{IEEEeqnarray}{LCl}
R_1 + R_{1p} + R_{1c} < I(U_1;X_1) + I(U_1,X_1;U_2,Y_1) - \delta(\epsilon).\label{eq:ach4}
\end{IEEEeqnarray} 

Finally, for $\msf{P}(\mc{E}_6)\rightarrow 0$ as $n\rightarrow\infty$, another application of the multivariate packing lemma as shown in \eqref{eq:multipack_3} gives the condition
\begin{IEEEeqnarray}{rCl}
R_1 + R_{1p} + R_{1c} + R_{2c} & < & I(U_1;X_1) + I(U_2;Y_1) + I(U_1,X_1;U_2,Y_1)-\delta(\epsilon).\label{eq:ach5}
\end{IEEEeqnarray}

A similar analysis leads to the following additional conditions for the probability of error in decoding $m_2$ to vanish as $n\rightarrow\infty$.
\begin{IEEEeqnarray}{rCl}
R_2 + R_{2p} & < & I(X_2;U_2,U_1,Y_2) - \delta(\epsilon),\label{eq:ach7}\\
R_2 + R_{2p} + R_{1c} & < & I(X_2,U_1;U_2,Y_2)- \delta(\epsilon),\label{eq:ach8}\\
R_2 + R_{2p} + R_{2c} & < & I(U_2;X_2) + I(U_2,X_2;U_1,Y_2)- \delta(\epsilon),\label{eq:ach9}\\
R_2 + R_{2p} + R_{2c} + R_{1c} & < & I(U_2;X_2) + I(U_1;Y_2) + I(U_2,X_2;U_1,Y_2)-\delta(\epsilon).\label{eq:ach10}
\end{IEEEeqnarray}

Hence the probability of error vanishes as $n\rightarrow\infty$ if the conditions \eqref{eq:ach1}-\eqref{eq:ach10} are satisfied. 
For the sake of brevity, let us first denote the RHS of the conditions \eqref{eq:ach1}-\eqref{eq:ach10} by $a,b,c,d,e,f,g,h,i,j$ respectively (ignoring the $\delta(\epsilon')$ and $\delta(\epsilon)$ terms). 

We then note the following relations among these terms which can be proved using the chain rule of mutual information, the Markov chains $U_1-X_1-(U_2,X_2,Y_1,Y_2)$ and $U_2-X_2-(U_1,X_1,Y_1,Y_2)$ and the independence of $(U_1,X_1)$ and $(U_2,X_2)$.
\begin{equation}\label{eq:relFM}
\begin{gathered}
e-a  \leq  \min\{c,d\},\\
f -a  \leq  d \leq f,\\
c  \leq  e  \leq  f,\\
i-b \leq  \min\{g,h\},\\
j-b  \leq  h \leq j,\\
g  \leq  i  \leq  j.
\end{gathered}
\end{equation}

We now employ Fourier-Motzkin elimination on the conditions \eqref{eq:ach1}-\eqref{eq:ach10} and $R_{1c},R_{1p},R_{2c},R_{2p}\geq 0$ to eliminate $R_{1c},R_{1p},R_{2c},R_{2p}$. The set of relations \eqref{eq:relFM} can be used to simplify this task by recognizing redundant constraints. At the end, we get the following achievable region:
\begin{equation}\label{eq:achreg1}
\begin{split}
R_1 & <  e-a,\\
R_2 & <  i-b,\\
R_1 + R_2 & <  c + j-a-b,\\
R_1 + R_2 & <  d + h-a-b,\\
R_1 + R_2 & <  f + g-a-b,\\
2R_1 + R_2 & <  c+h+f-2a-b,\\
R_1+2R_2 & < d +g+j-a-2b.
\end{split}
\end{equation}

Using the same facts as those used to prove \eqref{eq:relFM}, we can show that the above region is the same as the Han-Kobayashi region. For the sake of completeness, we show this explicitly. 
\begin{itemize}
\item Consider the upper bound on $R_1$:
\begin{IEEEeqnarray}{rCl}
e-a & = & I(U_1,X_1;U_2,Y_1)\nonumber\\
& \stackrel{(a)}{=} & I(X_1;U_2,Y_1)\nonumber\\
& \stackrel{(b)}{=} & I(X_1;Y_1|U_2),\label{eq:achreg2}
\end{IEEEeqnarray}where step $(a)$ follows since $U_1-X_1-(U_2,Y_1)$ is a Markov chain, and step $(b)$ follows since $X_1$ is independent of $U_2$.
\item Similarly, \begin{equation}\label{eq:achreg3}i-b= I(X_2;Y_2|U_1).\end{equation}
\item Consider the first upper bound on the sum-rate ${c+j-a-b}$:
\begin{IEEEeqnarray}{lCl}
c+j-a-b \nonumber\\
 =  I(X_1;U_1,U_2,Y_1) + I(U_2;X_2) + I(U_1;Y_2)  \nonumber\\
 \quad\quad +\> I(U_2,X_2;U_1,Y_2)- I(U_2;X_2)-I(U_1;X_1)\nonumber\\
 \stackrel{(a)}{=}  I(X_1;U_2,Y_1|U_1) + I(U_1;Y_2) + I(U_2,X_2;U_1,Y_2)\nonumber\\
 \stackrel{(b)}{=}  I(X_1;U_2,Y_1|U_1) + I(U_1;Y_2) + I(X_2;U_1,Y_2)\nonumber\\
 \stackrel{(c)}{=}  I(X_1;U_2,Y_1|U_1) + I(U_1;Y_2) + I(X_2;Y_2|U_1)\nonumber\\
 \stackrel{(d)}{=}  I(X_1;Y_1|U_1,U_2) + I(X_2,U_1;Y_2),\label{eq:achreg4}
\end{IEEEeqnarray}where step $(a)$ follows by the chain rule of mutual information, step $(b)$ follows by the Markov chain $U_2-X_2-(U_1,Y_2)$, step $(c)$ follows since $U_1$ and $X_2$ are independent and step $(d)$ follows by the independence of $U_2$ and $(U_1,X_1)$. 
\item By similar steps, $f+g-a-b =$  \begin{equation}\label{eq:achreg5}I(X_1,U_2;Y_1) + I(X_2;Y_2|U_1,U_2).\end{equation}
\item The remaining upper-bound on the sum-rate $d+h-a-b$ can be simplified as follows: \begin{IEEEeqnarray}{lCl}
d+h-a-b \nonumber\\
=  I(X_1,U_2;U_1,Y_1) + I(X_2,U_1;U_2,Y_2) \nonumber\\
\quad -\> I(U_1;X_1) - I(U_2;X_2)\nonumber\\
 =  I(X_1,U_2;Y_1|U_1) + I(X_2,U_1;Y_2|U_2),\label{eq:achreg6}
\end{IEEEeqnarray}which follows by the chain rule of mutual information and the independence of $(U_1,X_1)$ and $(U_2,X_2)$.
\item The upper bound on $2R_1+R_2$ can be simplified as follows:
\begin{IEEEeqnarray}{lCl}
c+h+f-2a-b \nonumber\\
= I(X_1;U_1,U_2,Y_1) + I(X_2,U_1;U_2,Y_2) + I(U_1,X_1) \nonumber\\
\quad  +\> I(U_2;Y_1) + I(U_1,X_1;U_2,Y_1) - 2I(U_1;X_1) -  I(U_2;X_2)\nonumber\\
 \stackrel{(a)}{=}  I(X_1;U_2,Y_1|U_1) + I(X_2,U_1;Y_2|U_2) + I(U_2;Y_1) + I(U_1,X_1;U_2,Y_1)\nonumber\\
 \stackrel{(b)}{=}  I(X_1;U_2,Y_1|U_1) + I(X_2,U_1;Y_2|U_2) + I(U_2;Y_1) + I(X_1;Y_1|U_2)\nonumber\\
 \stackrel{(c)}{=}  I(X_1;Y_1|U_1,U_2) + I(X_2,U_1;Y_2|U_2) + I(X_1,U_2;Y_1),\label{eq:achreg7}
\end{IEEEeqnarray}where step $(a)$ holds by the chain rule of mutual information and the independence of $U_1$ and $(U_2,X_2)$, step $(b)$ follows by $U_1-X_1-(U_2,Y_1)$ and the independence of $X_1$ and $U_2$, and step $(c)$ follows by the chain rule of mutual information and the independence of $U_2$ and $(U_1,X_1)$.
\item Finally, $d+g+j-a-2b$ can be similarly shown to be equal to \begin{equation}\label{eq:achreg8} I(X_2;Y_2|U_1,U_2) + I(X_1,U_2;Y_1|U_1) + I(X_2,U_1;Y_2).\end{equation}
\end{itemize}

From \eqref{eq:achreg1}-\eqref{eq:achreg8} and including a time-sharing random variable $Q$, we get that the following region is achievable:
\begin{equation}\label{eq:achreg}
\begin{split}
R_1 & < I(X_1;Y_1|U_2,Q),\\
R_2 & < I(X_2;Y_2|U_1,Q),\\
R_1 + R_2 & < I(X_1;Y_1|U_1,U_2,Q) +I(X_2,U_1;Y_2|Q) ,\\
R_1 + R_2 & < I(X_1,U_2;Y_1|U_1,Q) + I(X_2,U_1;Y_2|U_2,Q),\\
R_1 + R_2 & < I(X_1,U_2;Y_1|Q) + I(X_2;Y_2|U_1,U_2,Q),\\
2R_1 + R_2 & < I(X_1;Y_1|U_1,U_2,Q) + I(X_2,U_1;Y_2|U_2,Q) + I(X_1,U_2;Y_1|Q),\\
R_1 + 2R_2 & < I(X_2;Y_2|U_1,U_2,Q) + I(X_1,U_2;Y_1|U_1,Q) + I(X_2,U_1;Y_2|Q),
\end{split}
\end{equation}
for pmf $p(q)p(u_1,x_1|q)p(u_2,x_2|q).$ This region is identical to the region in \eqref{eq:achreg_prelim}.\hfill\IEEEQED

\section{State-dependent Interference channels}\label{sec:state}
In this section, we focus on the particular setup of the state-dependent Z-interference channel (S-D Z-IC) with noncausal state information at the interfering transmitter, as depicted in Fig.~\ref{fig:model_gen}. We provide a simple achievability scheme for this setup, that is obtained from the alternative achievability scheme for the general interference channel. This scheme is shown to be optimal for the deterministic case. The auxiliary random variable used for encoding at the interfering transmitter now implicitly captures some part of the message as well as some part of the state sequence realization.  
The achievability scheme can also be viewed as a generalization of the schemes presented in \cite{Cad09} and \cite{Dua13b}.


After characterizing the capacity region of the deterministic S-D Z-IC, we investigate a special case in detail: the modulo-additive S-D Z-IC. The modulo-additive channel is motivated by the linear deterministic model which has gained popularity over the recent years for studying wireless networks \cite{Ave11}. For this case (which can be thought of as a linear deterministic model with only one \emph{bit level}), we obtain an explicit description of the capacity region and furthermore, show that the capacity region is also achieved by the standard Gelfand-Pinsker coding over the first link and treating interference as noise over the second link. Following this, the modulo-additive S-D Z-IC with multiple levels is considered and some discussion is provided about the capacity region and the performance of simple achievability schemes.

To summarize, this section contains the following contributions:
\begin{itemize}
\item An achievable rate region for the S-D Z-IC,
\item Capacity region of the injective deterministic S-D Z-IC,
\item Modulo-additive S-D Z-IC: optimality of treating interference-as-noise and other properties.
\end{itemize}

\subsection{Results for the State-Dependent Channel}\label{subsec:main_res_state}

The following theorem provides an inner bound to the capacity region of the S-D Z-IC in Fig.~\ref{fig:model_gen}.
\begin{thm}\label{thm:gen_ach}
A rate pair $(R_1,R_2)$ is achievable for the channel in Fig.~\ref{fig:model_gen} if
\begin{equation}\label{eq:gen_ach}
\begin{split}
R_1 & < I(U;Y_1|Q)-I(U;S|Q),\\
R_2 & < I(X_2;Y_2|V,Q),\\
R_2 & < I(V,X_2;Y_2|Q) - I(V;S|Q),\\
R_1 + R_2 & < I(U;Y_1|Q)+ I(V,X_2;Y_2|Q)\\
 & \quad\quad -I(U;S|Q)- I(U,S;V|Q),
\end{split}
\end{equation}for some pmf $p(q)p(u,v|s,q)p(x_1|u,v,s,q)p(x_2|q).$
\end{thm}


For the injective deterministic S-D Z-IC, we can identify natural choices for the auxiliary random variables in Theorem~\ref{thm:gen_ach} that, in fact, yield the capacity region. This result is stated in the following theorem.

\begin{thm}\label{thm:cap}
The capacity region of the injective deterministic S-D Z-IC in Fig.~\ref{fig:model_state_det} is the set of rate pairs $(R_1,R_2)$ that satisfy
\begin{equation}\label{eq:cap}
\begin{split}
R_1 & \leq H(Y_1|S,Q),\\
R_2 & \leq  H(Y_2|T_1,Q),\\
R_2 & \leq  H(Y_2|Q) - I(T_1;S|Q),\\
R_1 + R_2 & \leq H(Y_1|T_1,S,Q)+H(Y_2|Q)-I(T_1;S|Q),
\end{split}
\end{equation}
for some pmf $p(q)p(x_1|s,q)p(x_2|q),$ where $|\mc{Q}|\leq 4$.
\end{thm}


\begin{remark} Note that the capacity region remains unchanged even if the first receiver is provided with the state information. The proof of this theorem is presented in subsection~\ref{subsec:proof_thm_cap}.
\end{remark}

\subsection{Proof of Theorem~\ref{thm:gen_ach}}\label{subsec:proofthm1}
Fix $p(u,v|s)p(x_1|u,v,s)p(x_2)$ and choose $0<\epsilon'<\epsilon$. 

\subsubsection*{Codebook Generation}
\begin{itemize}
\item Encoder 2 generates $2^{nR_2}$ codewords $x_2^n(m_2), m_2\in[1:2^{nR_2}]$ i.i.d. according to $p(x_2)$.
\item Encoder 1 generates $2^{n(R_1+R_1')}$ codewords $u^n(m_1,l_1)$ i.i.d. according to $p(u)$, where $m_1\in[1:2^{nR_1}]$ and $l_1\in[1:2^{nR_1'}]$. Encoder 1 also generates $2^{nR_2'}$ codewords $v^n(l_2), l_2\in[1:2^{nR_2'}]$ i.i.d. according to $p(v)$.
\end{itemize}

\subsubsection*{Encoding}
\begin{itemize}
\item To transmit message $m_2$, encoder~2 transmits $x_2^n(m_2)$.
\item Assume that the message to be transmitted by encoder~1 is $m_1$. After observing $s^n$, it finds a pair $(l_1,l_2)$ such that $(u^n(m_1,l_1),v^n(l_2),s^n)\in\mc{T}^{(n)}_{\epsilon'}$. Then it transmits $x_1^n$, which is generated i.i.d. according to $p(x_1|u,v,s)$.
\end{itemize}

\subsubsection*{Decoding}
\begin{itemize}
\item Decoder 1 finds a unique $\hat{m}_1$ such that $(u^n(\hat{m}_1,l_1),y_1^n)\in\mc{T}^{(n)}_{\epsilon}$ for some $l_1$.
\item Decoder 2 finds a unique $\hat{m}_2$ such that $(x_2^n(\hat{m}_2),v^n(l_2),y_2^n)\in\mc{T}_\epsilon^{(n)}$ for some $l_2$.
\end{itemize}

\subsubsection*{Probability of Error}
Due to the symmetry of the code, the average probability of error $\msf{P}(\mc{E})$ is equal to $\msf{P}(\mc{E}|M_1,M_2)$, so we can assume $(M_1,M_2) = (1,1)$ and analyze $\msf{P}(\mc{E}|1,1)$. Let $(L_1,L_2)$ denote the pair of indices chosen by encoder 1 such that $(U^n(1,L_1),V^n(L_2),S^n)\in\mc{T}^n_{\epsilon'}$. 

We now define events that cover the error event:
\begin{IEEEeqnarray*}{rCl}
\mc{E}_1 & \triangleq & \{(U^n(1,l_1),V^n(l_2),S^n)\notin\mc{T}^{(n)}_{\epsilon'} \text{ for all } l_1, l_2\}, \label{eq:E1}\\
\mc{E}_2 & \triangleq & \{(U^n(1,L_1),Y_1^n)\notin\mc{T}^{(n)}_{\epsilon}\}, \label{eq:E2}\\
\mc{E}_3 & \triangleq & \{(U^n(m_1,l_1),Y_1^n)\in\mc{T}^{(n)}_{\epsilon} \text{ for some }m_1\neq 1, l_1\}, \label{eq:E3}\\
\mc{E}_4 & \triangleq & \{(X_2^n(1),V^n(L_2),Y_2^n)\notin\mc{T}^{(n)}_{\epsilon}\} \label{eq:E4},\\
\mc{E}_5 & \triangleq & \{(X_2^n(m_2),V^n(l_2),Y_2^n)\in\mc{T}^{(n)}_{\epsilon} \text{ for some }m_2\neq 1, l_2\} \label{eq:E5}.
\end{IEEEeqnarray*}

Since an error occurs only if at least one of the above events occur, we have the following upper bound on the average probability of error:
$$\msf{P}(\mc{E}) \leq \msf{P}(\mc{E}_1) + \msf{P}(\mc{E}_2\cap\mc{E}_1^c) + \msf{P}(\mc{E}_3) + \msf{P}(\mc{E}_4\cap\mc{E}_1^c) + \msf{P}(\mc{E}_5).$$

Similar to the proof of the mutual covering lemma \cite[Ch. 8]{Gam12}, we can show that $\msf{P}(\mc{E}_1)\rightarrow 0$ as $n\rightarrow\infty$ if 
\begin{IEEEeqnarray}{rCl}
R_1' & > & I(U;S) + \delta(\epsilon'),\label{eq:ach1_state}\\
R_2' & > & I(V;S) + \delta(\epsilon'),\label{eq:ach2_state}\\
R_1' + R_2' & > & I(U;S) + I(U,S;V) + \delta(\epsilon')\label{eq:ach3_state},
\end{IEEEeqnarray}
where $\delta(\epsilon')\rightarrow 0$ as $\epsilon'\rightarrow 0.$

By the conditional typicality lemma \cite[Ch. 2]{Gam12}, $\msf{P}(\mc{E}_2\cap\mc{E}_1^c)$ and $\msf{P}(\mc{E}_4\cap\mc{E}_1^c)$ both tend to zero as $n\rightarrow\infty$.

By the packing lemma \cite[Ch. 3]{Gam12}, for $\msf{P}(\mc{E}_3)\rightarrow 0$, we require
\begin{equation}\label{eq:ach4_state} R_1+R_1' < I(U;Y_1) - \delta(\epsilon),\end{equation}
and for $\msf{P}(\mc{E}_5)\rightarrow 0$, we require
\begin{IEEEeqnarray}{rCl}
R_2 & < & I(X_2;Y_2|V) - \delta(\epsilon),\label{eq:ach5_state}\\
R_2 + R_2' & < & I(V,X_2;Y_2) - \delta(\epsilon),\label{eq:ach6_state}
\end{IEEEeqnarray}
where $\delta(\epsilon)\rightarrow 0$ as $\epsilon\rightarrow 0.$ Hence, $\msf{P}(\mc{E})\rightarrow 0$ if \eqref{eq:ach1_state}, \eqref{eq:ach2_state}, \eqref{eq:ach3_state}, \eqref{eq:ach4_state}, \eqref{eq:ach5_state}, \eqref{eq:ach6_state} are satisfied. 
Allowing coded-time sharing with a time-sharing random variable $Q$ and eliminating $R_1',R_2'$ via Fourier-Motzkin elimination, we obtain the region \eqref{eq:gen_ach}.\hfill\IEEEQED

\subsection{Proof of Theorem~\ref{thm:cap}}\label{subsec:proof_thm_cap}
Achievability follows from Theorem~\ref{thm:gen_ach} by choosing $U=Y_1$ and $V=T_1$. These choices are valid since encoder 1 knows $(M_1,S^n)$, which determines $T_1^n$ and $Y_1^n$. We now prove the converse.

Given a sequence of codes that achieves reliable communication (i.e. $P_e^{(n)}\rightarrow 0$ as $n\rightarrow\infty$) at rates $(R_1,R_2)$, we have, by Fano's inequality:
\begin{IEEEeqnarray*}{c}
H(M_1|Y_1^n) \leq n\epsilon_n,\\
H(M_2|Y_2^n) \leq n\epsilon_n,
\end{IEEEeqnarray*}
where $\epsilon_n\rightarrow 0$ as $n\rightarrow\infty.$

Using these, we can establish an upper bound on $R_1$ as follows,
\begin{IEEEeqnarray}{rCl}
nR_1 & = & H(M_1) \nonumber\\
& = & H(M_1|S^n) \nonumber\\
& \leq & I(M_1;Y_1^n|S^n) + n\epsilon_n \nonumber\\
& \leq & H(Y_1^n|S^n) + n\epsilon_n \nonumber\\
& \leq & \sum_{i=1}^n H(Y_{1i}|S_i) + n\epsilon_n. \nonumber
\end{IEEEeqnarray}

A simple upper bound on $R_2$ is established in the following:
\begin{IEEEeqnarray}{rCl}
nR_2 & = & H(M_2)\nonumber\\
& = & H(M_2|T_1^n)\nonumber\\
& \leq & I(M_2;Y_2^n|T_1^n) + n\epsilon_n\nonumber\\
& \leq & H(Y_2^n|T_1^n) + n\epsilon_n\nonumber\\
& \leq & \sum_{i=1}^n H(Y_{2i}|T_{1i}) + n\epsilon_n.\nonumber
\end{IEEEeqnarray}

For the second upper bound on $R_2$, consider the following:
\begin{IEEEeqnarray}{rCl}
nR_2 & = & H(M_2)\nonumber\\
& = & H(M_2) + H(Y_2^n|M_2) - H(Y_2^n|M_2) \nonumber\\
& = & H(Y_2^n) + H(M_2|Y_2^n) - H(Y_2^n|M_2) \nonumber\\
& \leq & \sum_{i=1}^nH(Y_{2i}) + n\epsilon_n - H(Y_2^n|M_2)\nonumber\\
& \stackrel{(a)}{=} & \sum_{i=1}^nH(Y_{2i}) + n\epsilon_n - H(T_1^n|M_2)\nonumber\\
& \stackrel{(b)}{=} & \sum_{i=1}^nH(Y_{2i}) + n\epsilon_n - H(T_1^n)\nonumber\\
& \leq & \sum_{i=1}^nH(Y_{2i}) + n\epsilon_n - I(T_1^n;S^n)\nonumber\\
& = & \sum_{i=1}^nH(Y_{2i}) + n\epsilon_n - H(S^n) + H(T_1^n|S^n)\nonumber\\
& \leq & \sum_{i=1}^nH(Y_{2i}) + n\epsilon_n - H(S^n) + \sum_{i=1}^nH(T_{1i}|S_i)\nonumber\\
& \stackrel{(c)}{=} & \sum_{i=1}^nH(Y_{2i}) + n\epsilon_n - \sum_{i=1}^nH(S_i) + \sum_{i=1}^nH(T_{1i}|S_i)\nonumber\\
& = & \sum_{i=1}^nH(Y_{2i}) + n\epsilon_n - \sum_{i=1}^nI(T_{1i};S_i)\nonumber\\
\end{IEEEeqnarray}
where step $(a)$ follows by the injectivity property, step $(b)$ follows because $T_1^n$ is independent of $M_2$, and step $(c)$ follows because $S^n$ is an i.i.d. sequence.


We now establish an upper bound on the sum-rate.

 \begin{IEEEeqnarray}{rL}
n(R_1+R_2) & = H(M_1|S^n) + H(M_2)\nonumber\\
& \leq I(M_1;T_1^n,Y_1^n|S^n) + n\epsilon_n + H(Y_2^n) + H(M_2|Y_2^n) - H(Y_2^n|M_2)\nonumber\\
& \leq I(M_1;T_1^n,Y_1^n|S^n) + n\epsilon_n + H(Y_2^n) + n\epsilon_n - H(Y_2^n|M_2)\nonumber\\
& \stackrel{(a)}{\leq} H(T_1^n,Y_1^n|S^n) + H(Y_2^n) - H(T_1^n|M_2)+ 2n\epsilon_n\nonumber\\
& \stackrel{(b)}{=} H(T_1^n,Y_1^n|S^n) + H(Y_2^n) - H(T_1^n)+ 2n\epsilon_n\nonumber\\
& = H(Y_1^n|S^n,T_1^n) + H(Y_2^n) - I(T_1^n;S^n)+ 2n\epsilon_n\nonumber\\
& \stackrel{(c)}{\leq} \sum_{i=1}^n H(Y_{1i}|S_i,T_{1i}) + \sum_{i=1}^n H(Y_{2i}) - \sum_{i=1}^nI(T_{1i};S_i)+ 2n\epsilon_n\nonumber
\end{IEEEeqnarray}
where as before, steps $(a)$, $(b)$ and $(c)$ follow because of injectivity property, independence of $T_1^n$ and $M_2$, and i.i.d. state respectively.

From the four bounds established in this section, we can complete the converse by introducing an independent time-sharing random variable $Q$ uniformly distributed on $[1:n]$ and defining $X_1$, $T_1$, $S$, $X_2$, $Y_1$, $Y_2$ to be $X_{1Q}$, $T_{1Q}$, $S_{Q}$, $X_{2Q}$, $Y_{1Q}$, $Y_{2Q}$ respectively. \hfill\IEEEQED

\subsection{Example: Modulo-Additive State-Dependent Z-Interference Channel}\label{subsec:modulo}

\begin{thm}\label{thm:modulo}
The capacity region of the modulo-additive S-D Z-IC in Fig.~\ref{fig:model_modulo} is given by the convex closure of the rate pairs $(R_1,R_2)$ satisfying
\begin{equation}\label{eq:cap_modulo}
\begin{split}
R_1 & < (1-\lambda)\log |\mc{X}| + \lambda H(\bm{p}),\\
R_2 & < \log |\mc{X}| - H\left(\lambda \bm{p} + (1-\lambda)\bm{\delta}_0\right),
\end{split}
\end{equation}
for some $\bm{p}\in\mc{P}_{\mc{X}}$, where $\mc{P}_{\mc{X}}$ denotes the probability simplex corresponding to $\mc{X}$, $H(\bm{p})$ stands for the entropy of the pmf $\bm{p}$ and $\bm{\delta}_0$ denotes the pmf that has unit mass at $0$.
\end{thm}

The capacity region when $\mc{X}=\{0,1\}$ and $S$ is i.i.d. Ber$\left(\frac{1}{2}\right)$ is shown in Figure~\ref{fig:modulo}.


\subsection*{Proof of Theorem~\ref{thm:modulo}}


\begin{figure}[!t]
\centering
\input{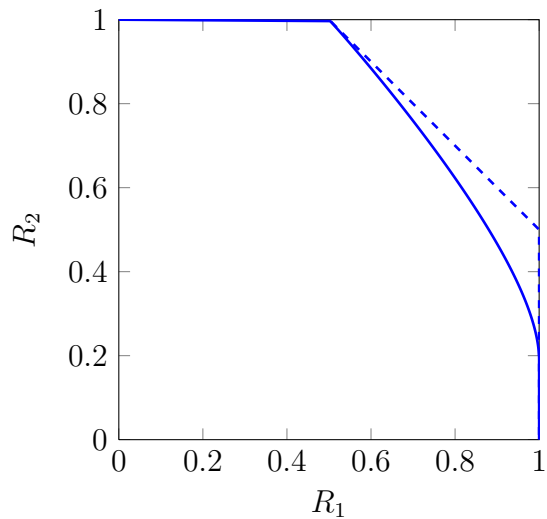}
\caption{Capacity Region with $\mc{X}=\{0,1\}$ and $S$ i.i.d. Ber$\left(\frac{1}{2}\right).$ The dotted line shows the capacity region when all nodes have state information. Note that the maximal sum-rate of $1.5$ bits/channel use is achievable with state information only at the interfering Tx.}
\label{fig:modulo}
\end{figure}

Consider the capacity region stated in Theorem~\ref{thm:cap}. Let $\bm{p}_{1,0}$, $\bm{p}_{1,1}$ and $\bm{p}_{2}$, all in $\mc{P}_{\mc{X}}$, be used to denote the pmf's ${p(x_1|s=0,q)}$, $p(x_1|s=1,q)$ and $p(x_2|q)$ respectively. Evaluating each of the constraints in \eqref{eq:cap} gives us the following expression for the capacity region:
\begin{equation}\label{eq:cap_modulo_1}
\begin{split}
R_1 & < (1-\lambda)H(\bm{p}_{1,0}) + \lambda H(\bm{p}_{1,1}),\\
R_2 & < H(\bm{p}_{2}),\\ 
R_2 & < H\left((1-\lambda)\bm{p}_{2} + \lambda\widetilde{\bm{p}}\right) + \lambda H(\bm{p}_{1,1}) \\
&\quad\quad - H\left(\lambda \bm{p}_{1,1} + (1-\lambda)\bm{\delta}_0\right),\\
R_1 + R_2 & < (1-\lambda)H(\bm{p}_{1,0})+H\left((1-\lambda)\bm{p}_{2} + \lambda\widetilde{\bm{p}}\right)\\
&\quad\quad + \lambda H(\bm{p}_{1,1}) - H\left(\lambda \bm{p}_{1,1} + (1-\lambda)\bm{\delta}_0\right),
\end{split}
\end{equation}
where $\widetilde{\bm{p}}\in\mc{P}_{\mc{X}}$ is a pmf that is defined as 
$$\widetilde{\bm{p}}(k) = \sum_{i=0}^{|\mc{X}|-1} \bm{p}_{1,1}(i)\bm{p}_{2}(k-i),\quad 0\leq k\leq |\mc{X}|-1,$$ and $k-i$ should be understood to be $(k-i)\text{ mod } |\mc{X}|$.

Firstly, we note that $\bm{p}_{1,0}$ should be chosen as the pmf of the uniform distribution to maximize $H(\bm{p}_{1,0})$, thus maximizing the RHS of the constraints in \eqref{eq:cap_modulo_1}. Similarly, $\bm{p}_2$ should also be chosen to be the pmf of the uniform distribution. Then, we can also remove the first constraint on $R_2$, since it is rendered redundant by the other constraint on $R_2$.
Thus, the capacity region is given by the convex closure of $(R_1,R_2)$ satisfying
\begin{equation}\label{eq:cap_modulo_3}
\begin{split}
R_1 & < (1-\lambda)\log(|\mc{X}|) + \lambda H(\bm{p}_{1,1}),\\
R_2 & < \log(|\mc{X}|) + \lambda H(\bm{p}_{1,1}) - H\left(\lambda \bm{p}_{1,1} + (1-\lambda)\bm{\delta}_0\right),\\
R_1 + R_2 & < (2-\lambda)\log(|\mc{X}|)+ \lambda H(\bm{p}_{1,1})\\
&\quad\quad\quad\quad - H\left(\lambda \bm{p}_{1,1} + (1-\lambda)\bm{\delta}_0\right),
\end{split}
\end{equation} for $\bm{p}_{1,1}\in\mc{P}_{\mc{X}}.$

For any $\bm{p}$, the region in \eqref{eq:cap_modulo} is contained in the region in \eqref{eq:cap_modulo_3} for $\bm{p}_{1,1}=\bm{p}$. Hence, the convex closure of \eqref{eq:cap_modulo} is contained in the convex closure of \eqref{eq:cap_modulo_3}.

However, also note that the region in \eqref{eq:cap_modulo_3} for any $\bm{p}_{1,1}$ is contained in the convex hull of two regions, one obtained by setting $\bm{p} = \bm{p}_{1,1}$ in \eqref{eq:cap_modulo} and the other obtained by setting $\bm{p}=\bm{\delta}_0$ in \eqref{eq:cap_modulo}. Hence, the convex closure of \eqref{eq:cap_modulo_3} is also contained in the convex closure of \eqref{eq:cap_modulo}. This concludes the proof of Theorem~\ref{thm:modulo}.\hfill\IEEEQED

\begin{remark}
The optimal sum-rate $(2-\lambda)\log |\mc{X}|$ is achieved by choosing $\bm{p}=\bm{\delta}_0$. This corresponds to setting the transmitted symbols of the first transmitter to $0$ when $S=1$ so that it does not interfere with the second transmission. The first transmitter then treats these symbols as stuck to $0$ and performs Gelfand-Pinsker coding. The second transmitter transmits at rate $\log(|\mc{X}|)$ bits/channel use. It can be easily verified that this is also the optimal sum-rate when all nodes are provided with the state~information. Thus, for this channel, the sum-capacity when all nodes have state information is the same as that when only encoder~1 has state information.
\end{remark}

\begin{remark}
Finally, we note that there is also another way to achieve the capacity region of the modulo additive S-D Z-IC. For this, first recall that to get the capacity region expression in Theorem~\ref{thm:cap}, we set the auxiliary random variables $U$ and $V$ in the expression in Theorem~\ref{thm:gen_ach} to $Y_1$ and $T_1$ respectively. Another choice, which corresponds to standard Gelfand-Pinsker coding for the first transmitter-receiver pair and treating interference as noise at the second receiver is to choose $V=\phi$ in Theorem~\ref{thm:gen_ach}. This gives us the following achievable region:
\begin{equation}\label{eq:int_noise}
\begin{split}
R_1 & < I(U;Y_1|Q)-I(U;S|Q),\\
R_2 & < I(X_2;Y_2|Q),
\end{split}
\end{equation} for some pmf $p(q)p(u|s,q)p(x_1|u,s,q)p(x_2|q)$. We can now see that for the modulo-additive S-D Z-IC, the capacity region is also achieved by making the following choices in the above region: $p(u|s=0)$ to be the uniform pmf over $\mc{X}$, $p(u|s=1)$ to be $\bm{p}$, $p(x_1|u,s)$ to be $\bm{\delta}_u$ (i.e. $X_1=U$) and $p(x_2)$ to be the uniform pmf over $\mc{X}$. Thus, the capacity region of the modulo-additive S-D Z-IC can also be achieved by treating interference as noise at the second receiver.
\end{remark}

\subsection{Multiple-level modulo-additive S-D Z-IC}\label{subsec:multiplelevel}

The linear deterministic model introduced in \cite{Ave11} consists of multiple \emph{bit levels} that roughly correspond to bits communicated at different power levels. The modulo-additive S-D Z-IC that we looked at in the previous subsection is a special case in which the number of levels is one. Extending the model to have multiple bit levels raises some interesting questions which we consider in this subsection. 

More specifically, consider the model depicted in Fig.~\ref{fig:model_modulo_multiple}, which can be thought of as three copies of the model in Fig.~\ref{fig:model_modulo}, which are however related by the common state affecting them. For simplicity, we restrict attention to the case when the alphabet on each level, denoted by $\mc{X}$, is the binary alphabet, i.e. $\{0,1\}$, and the state is Ber$(0.5).$ Let $L$ denote the number of bit levels.

\begin{figure}[!th]
\centering
\includegraphics[scale=1.5]{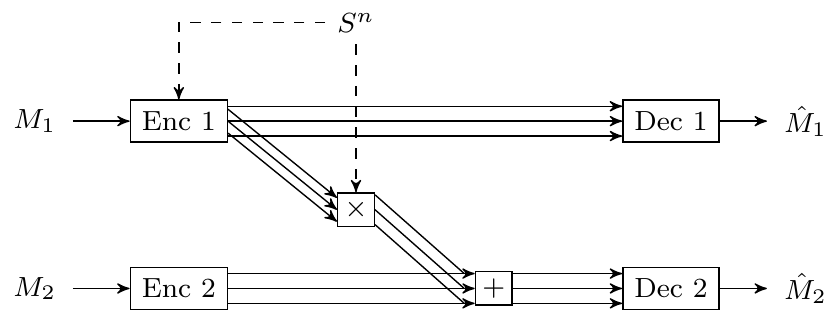}
\caption{The Modulo-Additive S-D Z-IC wit multiple bit levels.}
\label{fig:model_modulo_multiple}
\end{figure}

\begin{figure}[!th]
\centering
\begin{subfigure}
\centering
 \resizebox{.5\linewidth}{!}{\input{twolevel_binary_SD_ZIC.tikz}}
\caption{Comparison of the different rate regions for 2-level binary modulo-additive S-D Z-IC}
\label{fig:modulo_multiple_2}
\end{subfigure}
\vspace{4mm}
\begin{subfigure}
\centering
\resizebox{.5\linewidth}{!}{\input{threelevel_binary_SD_ZIC.tikz}}
\caption{Comparison of the different rate regions for 3-level binary modulo-additive S-D Z-IC}
\label{fig:modulo_multiple_3}
\end{subfigure}
\end{figure}

This model also falls under the injective-deterministic setup for which we have completely characterized the capacity region. So the capacity region can be easily computed, as we indeed do in the following. This evaluation also allows us to immediately compare the capacity region with the rates achieved by some straightforward achievability schemes that we can employ. In particular, consider the following two simple achievability schemes:
\begin{itemize}
\item ``Separation'': The simplest strategy one can employ is to separately consider each level and communicate over it independently of the other levels. This gives us that the rate pairs $(R_1,R_2)$ satisfying
\begin{equation}\label{eq:ach_separation}
\begin{split}
R_1 & < \frac{L}{2} + \sum_{i=1}^{L}\frac{1}{2} H(\bm{p}_i),\\
R_2 & < L - \sum_{i=1}^{L}H\left( \bm{p}_i + \bm{\delta}_0\right),
\end{split}
\end{equation}
for some $\bm{p}_1,\bm{p}_2,\dots,\bm{p}_L\in\mc{P}_{\mc{X}}$ are achievable.
\item ``Communicate state'': Alternatively, by noticing that strictly better rates could have been achieved if decoder~2 also had access to the state information, we can reserve one level to communicate the state from encoder~1 to decoder~2. This is done by ensuring that encoder~1 transmits a 1 on this reserved level whenever the state is 1, and encoder~2 constantly transmits a 0 on this level. The nodes communicate on the remaining levels keeping in mind that now decoder~2 also has state information. Note that while no communication can happen between encoder~2 and decoder~2 on the reserved level, encoder~1 can still communicate with decoder~1 at rate 0.5 on this level by treating it as a channel with stuck bits (bit equals 1 whenever state equals 1). This strategy provides us the following achievable region:
\begin{equation}\label{eq:ach_state}
\begin{split}
R_1 & < \frac{L}{2} + \frac{1}{2} H(\bm{p}),\\
R_2 & < L-1 - \frac{1}{2}H\left( \bm{p}\right),
\end{split}
\end{equation}
for some $\bm{p}\in\mc{P}_{\mc{X}^{L-1}}$.
\end{itemize}

We can expect that the suboptimality of reserving one level for communicating the state should become relatively small as the number of levels increases i.e. at high SNR. This is corroborated by the numerical analysis, shown in Figs.~\ref{fig:modulo_multiple_2} and \ref{fig:modulo_multiple_3}, in which we can see that there is a marked improvement in the rates achieved by this scheme relative to the capacity region as we increase the number of levels from 2 to 3. Indeed, since all the levels are affected by the same state, the entropy of the state becomes small compared to the communication rates as the SNR increases, so it is not a big overhead to explicitly communicate the state to decoder~2 at high SNR. However, at low SNR, the figures show that the overhead incurred is quite high due to which this approach is significantly suboptimal, while the simple scheme of treating the levels separately results in achieving very close to the entire capacity region.

\section{Interference Channels with Partial Cribbing}\label{sec:cribbing}

In this section, we focus on deterministic Z-interference channels when the interfering transmitter can overhear the signal transmitted by the other transmitter after it passes through some channel. This channel is also modeled as a deterministic channel, dubbed as \emph{partial cribbing} in \cite{Asn13}. Deterministic models, in particular linear deterministic models \cite{Ave11}, have gained popularity due to the observation that they are simpler to analyze and are provably close in performance to Gaussian models. 

There have been quite a few very sophisticated achievability schemes designed for interference channels with causal cribbing encoders, however optimality of the achievable rate regions has not been addressed. In the most general interference channel model with causal cribbing \cite{Yan11}, each encoder needs to split its message into four parts: a common part to be sent cooperatively, a common part to be sent non-cooperatively, a private part to be sent cooperatively and a private part to be sent non-cooperatively. Further, because of the causal nature of cribbing, achievability schemes usually involve block-Markov coding, so that each encoder also needs to consider the cooperative messages of both encoders from the previous block. Motivated by the alternative achievability scheme we have presented earlier for the general interference channel, we present a simple optimal achievability scheme that minimizes the rate-splitting that is required. Specifically, while encoder~2 only splits its message into a cooperative and non-cooperative private part, encoder~1 does not perform any rate-splitting at all. By focusing on the specific configuration of the Z-interference channel, we are able to prove the optimality of an achievability scheme that is simpler than the highly involved achievability schemes for the general case that are currently known. 

\subsection{Result for Partial Cribbing}\label{subsec:main_res_crib}
\begin{thm}\label{thm:part_crib}
The capacity region of the injective deterministic Z-interference channel with unidirectional partial cribbing, depicted in Fig.~\ref{fig:model_crib}, is given by the convex closure of $(R_1,R_2)$ satisfying \begin{equation}\label{eq:cap_part_crib}
\begin{split}
R_1 & \leq H(Y_1|W),\\
R_2 & \leq \min\Big(H(Y_2), H(Y_2,Z_2|T_1,W)\Big),\\
R_1 + R_2 & \leq H(Y_1|T_1,W)+\min\Big(H(Y_2), H(Y_2,Z_2|W)\Big),
\end{split}
\end{equation}
for $p(w)p(x_1|w)p(x_2|w),$ where $W$ is an auxiliary random variable whose cardinality can be bounded as
$|\mathcal{W}|\leq |\mathcal{Y}_2|+3.$
\end{thm}

The proof of this theorem is presented below.

\subsection{Proof of Theorem~\ref{thm:part_crib}}\label{subsec:proof_crib}
\emph{Achievability}\\
Choose a pmf $p(w)p(u_d,u_c,x_1|w)p(x_2,z_2|w)$ and ${0<\epsilon'<\epsilon}$, where for the sake of generality, we use the auxiliary random variables $U_d$ and $U_c$. In the injective deterministic case at hand, they can be set to $Y_1$ and $T_1$ respectively.
\subsubsection*{Codebook Generation}
The communication time is divided into $B$ blocks, each containing $n$ channel uses, and an independent random code is generated for each block $b\in[1:B]$. Whenever it is clear from the context, we suppress the dependence of codewords on $b$ to keep the notation simple. The messages in block $B$ are fixed apriori, so a total of $B-1$ messages are communicated over the $B$ blocks. The resulting rate loss can be made as negligible as desired by choosing a sufficiently large $B$.

We split $R_2$ as $R_2'+R_2''$, which corresponds to the split of message~2 into two parts, one that will be sent cooperatively by both transmitters to receiver~2 and the other non-cooperatively only by transmitter~2 to receiver~2. For each block $b$, let $m_{2,b}'\in[1:2^{nR_2'}]$ and $m_{2,b}''\in[1:2^{nR_2''}]$. For each block $b\in [1:B]$, we generate $2^{nR_2'}$ sequences $w^n$ i.i.d. according to $p(w)$.
\begin{itemize}
\item For each $w^n$ in block $b$, we generate $2^{nR_{2}'}$ sequences $\left\{z_2^n(w^n,m'_{2,b})\right\}$ i.i.d. according to $p(z_2|w)$. Then for each $(w^n,z_2^n)$, we generate $2^{nR_2''}$ sequences $\left\{x_2^n(w^n,z_2^n,m''_{2,b})\right\}$ i.i.d. according to $p(x_2|z_2,w)$.
\item For each $w^n$ in block $b$, we generate $2^{nR_c}$ sequences $\left\{u_c^n(w^n,l_c)\right\}$ i.i.d. according to $p(u_c|w)$, where $l_c\in[1:2^{nR_c}]$. We also generate $2^{n(R_1+R_d)}$ sequences $\left\{u_d^n(m_{1,b},l_d)\right\}$ i.i.d. according to $p(u_{d})$, where $m_{1,b}\in[1:2^{nR_1}]$ and $l_d\in[1:2^{nR_d}]$. \footnote{Note that the $u_d^n$ sequences are generated independently of the $w^n$ sequences.}
\end{itemize}

\subsubsection*{Encoding}
Let us assume for now that as a result of the cribbing, encoder 1 knows $m'_{2,b-1}$ at the end of block ${b-1}$. Then in block $b$, both encoders can encode $m'_{2,b-1}$ using $w^n(m'_{2,b-1})$ where $w^n$ is from the code for block $b$.
\begin{itemize}
\item To transmit message $m_{1,b}$, encoder 1 finds a pair $(l_{d},l_{c})$ such that $$(w^n(m'_{2,b-1}),u_c^n(w^n,l_c),u_d^n(m_{1,b},l_{d}))\in\mc{T}^{(n)}_{\epsilon'}.$$ It transmits $x_1^n$ that is generated i.i.d. according to $p(x_1|w,u_d,u_c)$.
\item To transmit message $m_{2,b}=(m'_{2,b},m''_{2,b})$, encoder 2 encodes $m'_{2,b}$ as $z_2^n(w^n,m'_{2,b})$ and then transmits $x_2^n(w^n,z_2^n,m''_{2,b})$.
\end{itemize}
We fix apriori the messages in block $B$ to be $m_{1,B}=1$, $m'_{2,B}=1$ and $m''_{2,B}=1$. Also, to avoid mentioning edge cases explicitly, whenever $m_{1,0}$, $m'_{2,0}$ or $m''_{2,0}$ appear, we assume that all are fixed to 1.

\subsubsection*{Decoding}
\begin{itemize}
\item \emph{Encoder~1:} At the end of block $b$, assuming it has already decoded $m'_{2,b-1}$ at the end of block $b-1$, encoder 1 decodes $m'_{2,b}$ by finding the unique $\hat{m}'_{2,b}$ such that the sequence $z_2^n$ it has observed via cribbing is equal to $z_2^n(w^n,\hat{m}'_{2,b})$.
\item \emph{Decoder~1:} In each block $b$, decoder 1 finds the unique $\hat{m}_{1,b}$ such that $(u_d^n(m_{1,b},l_{d}),y_1^n)\in\mc{T}^{(n)}_{\epsilon}$ for some $l_{d}$.
\item \emph{Decoder~2:} Decoder 2 performs backward decoding as follows: 
\begin{itemize}
\item In block $B$, decoder 2 finds a unique $m'_{2,B-1}$ such that the condition \eqref{eq:jtd_dec2_B} is satisfied for some $l_c.$
\begin{equation}\label{eq:jtd_dec2_B}
(w^n(\hat{m}'_{2,B-1}),z_2^n(w^n,1),x_2^n(w^n,z_2^n,1),u_c^n(w^n,l_c),y_2^n)\in\mc{T}^{(n)}_{\epsilon}
\end{equation}
\item In block $b$, assuming $m'_{2,b}$ has been decoded correctly, it finds the unique $(\hat{m}'_{2,b-1}, \hat{m}''_{2,b})$ such that  the condition \eqref{eq:jtd_dec2} is satisfied for some $l_{c}$.
\begin{equation}\label{eq:jtd_dec2}
(w^n(\hat{m}'_{2,b-1}),z_2^n(w^n,m'_{2,b}),x_2^n(w^n,z_2^n,\hat{m}''_{2,b}),u_c^n(w^n,l_c),y_2^n)\in\mc{T}^{(n)}_{\epsilon}
\end{equation}
\end{itemize}
\end{itemize}

\subsubsection*{Probability of Error}

To get a vanishing probability of error, we can impose the conditions described in the following list.
\begin{itemize}
\item 
Similar to the proof of the mutual covering lemma \cite[Ch. 8]{Gam12}, we can show that the following conditions are sufficient for the success of encoding at the first transmitter:
\begin{IEEEeqnarray}{rCl}
R_d & > & I(U_d;W) +\delta(\epsilon'),\label{eq:dec1}\\
R_d + R_c & > & I(U_d;U_c,W)+\delta(\epsilon')\label{eq:dec2}.
\end{IEEEeqnarray}
\item For the decoding at encoder 1 to succeed:
\begin{equation}
R'_2 < H(Z_2|W)-\delta(\epsilon).\label{eq:dec3}
\end{equation}
\item For decoding at decoder 1 to succeed:
\begin{equation}
R_1 + R_d < I(U_d;Y_1)-\delta(\epsilon).\label{eq:dec4}
\end{equation}
\item For the backward decoding at decoder 2 to succeed, it is sufficient that the following conditions are satisfied:
\begin{IEEEeqnarray}{rCl}
R''_2 & < & I(X_2;Y_2|W,U_c,Z_2)-\delta(\epsilon),\label{eq:dec5}\\
R''_2+ R_c & < & I(U_c,X_2;Y_2|W,Z_2)-\delta(\epsilon),\label{eq:dec6}\\
R'_2 + R''_2 + R_c & < &  I(W,U_c,X_2;Y_2)-\delta(\epsilon).\label{eq:dec7}
\end{IEEEeqnarray}
\end{itemize}

Noting that $R'_2+R''_2=R_2$, eliminating $(R_d,R_c,R'_2,R''_2)$ from \eqref{eq:dec1}-\eqref{eq:dec7} via Fourier-Motzkin elimination, and substituting $U_d=Y_1$ and $U_c=T_1$, we get the achievable region in \eqref{eq:cap_part_crib} with the following additional bound on $R_1$:
$$R_1< H(Y_1|W,T_1) + H(Y_2|W,Z_2).$$
To conclude the proof of achievability, we show that this bound is rendered redundant by $R_1<H(Y_1|W)$ which can be proved by the following chain of inequalities:
\begin{IEEEeqnarray*}{rCl}
H(Y_1|W,T_1) + H(Y_2|W,Z_2) & \geq & H(Y_1|W,T_1) + H(Y_2|W,X_2)\\
& = & H(Y_1|W,T_1) + H(T_1|W,X_2)\\
& = & H(Y_1|W,T_1) + H(T_1|W)\\
& = & H(Y_1,T_1|W)\\
& \geq & H(Y_1|W).
\end{IEEEeqnarray*}

\emph{Converse}\\
We now establish the converse. By Fano's inequality, we have the following two relations that are satisfied by any sequence of codes that achieve reliable communication:
$$H(M_1|Y_1^n)\leq n\epsilon_n,\quad H(M_2|Y_2^n)\leq n\epsilon_n,$$
where $\epsilon_n\rightarrow 0$ as $n\rightarrow\infty.$

First, an upper bound on $R_1$ is established in \eqref{eq:conv1}.
\begin{IEEEeqnarray}{rCl}
nR_1 & = & H(M_1)\nonumber\\
& = & H(M_1|Z_2^n)\nonumber\\
& \leq & I(M_1;Y_1^n|Z_2^n) + n\epsilon_n\nonumber\\
& \leq & H(Y_1^n|Z_2^n) + n\epsilon_n\nonumber\\
& \leq & \sum_{i=1}^n H(Y_{1i}|Z_2^{i-1}) + n\epsilon_n\nonumber\\
& = & \sum_{i=1}^n H(Y_{1i}|W_i) + n\epsilon_n,\IEEEeqnarraynumspace\label{eq:conv1}
\end{IEEEeqnarray}
where $W_i\triangleq Z_2^{i-1}$.


Next, we establish two bounds on $R_2$, the first one in \eqref{eq:conv2} as follows:
\begin{IEEEeqnarray}{rCl}
nR_2 & = & H(M_2)\nonumber\\
& \leq & I(M_2;Y_2^n) + n\epsilon_n\nonumber\\
& \leq & H(Y_2^n) + n\epsilon_n\nonumber\\
& \leq & \sum_{i=1}^n H(Y_{2i}) + n\epsilon_n,\IEEEeqnarraynumspace\label{eq:conv2}
\end{IEEEeqnarray}
and the second one in \eqref{eq:conv3} below:
\begin{IEEEeqnarray}{rCl}
nR_2 & = & H(M_2|M_1)\nonumber\\
& = & H(M_2,Z_2^n|M_1)\nonumber\\
& = & H(Z_2^n|M_1) + H(M_2|M_1,Z_2^n)\nonumber\\
& \stackrel{(a)}{=} & H(Z_2^n|M_1) + H(M_2|M_1,Z_2^n,T_1^n)\nonumber\\
& \leq & H(Z_2^n|M_1) + I(M_2;Y_2^n|M_1,Z_2^n,T_1^n) + n\epsilon_n\nonumber\\
& \leq & H(Z_2^n) + H(Y_2^n|Z_2^n,T_1^n) + n\epsilon_n\nonumber\\
& \leq & \sum_{i=1}^n H(Z_{2i}|W_{i}) + \sum_{i=1}^n H(Y_{2i}|W_i,T_{1i},Z_{2i}) + n\epsilon_n,\nonumber\\
& \stackrel{(b)}{=} & \sum_{i=1}^n H(Z_{2i}|W_{i},T_{1i}) + \sum_{i=1}^n H(Y_{2i}|W_i,T_{1i},Z_{2i}) + n\epsilon_n,\nonumber\\
& = & \sum_{i=1}^n H(Y_{2i},Z_{2i}|W_{i},T_{1i}) + n\epsilon_n,\IEEEeqnarraynumspace\label{eq:conv3}
\end{IEEEeqnarray}
where step $(a)$ follows because $T_1^n$ is a function of $X_1^n$ which is a function of $(M_1,Z_2^n)$, and step $(b)$ follows because $Z_{2i}-W_{i}-T_{1i}$.

Finally, we establish two bounds on the sum-rate, the first one in \eqref{eq:conv4} below:
\begin{IEEEeqnarray}{rRl}
& n(R_1+R_2) \nonumber\\
& = & H(M_1|Z_1^n) + H(M_2,Z_2^n)\nonumber\\
& = & H(M_1,T_1^n|Z_2^n) + H(Z_2^n) + H(M_2|Z_2^n)\nonumber\\
& = & H(T_1^n|Z_2^n) + H(M_1|T_1^n,Z_2^n)+ H(Z_2^n)  + H(M_2|Z_2^n)\nonumber\\
& \stackrel{(a)}{\leq} & H(Y_2^n|X_2^n,Z_2^n) + I(M_1;Y_1^n|T_1^n,Z_2^n) + H(Z_2^n) + I(M_2;Y_2^n|Z_2^n) + n\epsilon_n\nonumber\\
& \leq & H(Y_2^n|X_2^n,Z_2^n) + H(Y_1^n|T_1^n,Z_2^n) + H(Z_2^n)+ H(Y_2^n|Z_2^n) - H(Y_2^n|M_2,Z_2^n)+ n\epsilon_n\nonumber\\
& = & H(Y_2^n|X_2^n,Z_2^n) + H(Y_1^n|T_1^n,Z_2^n) + H(Z_2^n)+ H(Y_2^n|Z_2^n) - H(Y_2^n|M_2,X_2^n,Z_2^n)+ n\epsilon_n\nonumber\\
& \stackrel{(b)}{=} & H(Y_2^n|X_2^n,Z_2^n) + H(Y_1^n|T_1^n,Z_2^n) + H(Z_2^n)+ H(Y_2^n|Z_2^n) - H(Y_2^n|X_2^n,Z_2^n)+ n\epsilon_n\nonumber\\
& = & H(Y_1^n|T_1^n,Z_2^n) + H(Z_2^n)+ H(Y_2^n|Z_2^n) + n\epsilon_n\nonumber\\
& = & H(Y_1^n|T_1^n,Z_2^n) + H(Y_2^n,Z_2^n)+ n\epsilon_n\nonumber\\
& \leq & \sum_{i=1}^n H(Y_{1i}|T_{1i},W_i) + \sum_{i=1}^n H(Y_{2i},Z_{2i}|W_i) + n\epsilon_n,\nonumber\\*\label{eq:conv4}
\end{IEEEeqnarray}
where step $(a)$ uses the fact that $H(T_1^n|Z_2^n)\leq H(Y_2^n|X_2^n,Z_2^n)$, which is proved below: \begin{IEEEeqnarray*}{rCl}
H(T_1^n|Z_2^n) & \leq & H(T_1^n)\\
& = & H(Y_2^n|X_2^n)\\
& = & H(Y_2^n|X_2^n,Z_2^n),
\end{IEEEeqnarray*}
and step $(b)$ follows because $M_2 - (X_2^n,Z_2^n) - Y_2^n$.

The second bound on the sum-rate is established in \eqref{eq:conv5} as follows:
\begin{IEEEeqnarray}{rRl}
& n(R_1+R_2) \nonumber\\
& = & H(M_1,T_1^n|Z_2^n) + H(M_2)\nonumber\\
& \leq & H(T_1^n|Z_2^n) + H(M_1|T_1^n,Z_2^n) + I(M_2;Y_2^n) + n\epsilon_n\nonumber\\
& \leq & H(Y_2^n|X_2^n,Z_2^n) + I(M_1;Y_1^n|T_1^n,Z_2^n) + H(Y_2^n)- H(Y_2^n|M_2) + n\epsilon_n\nonumber\\
& \leq & H(Y_2^n|X_2^n,Z_2^n) + H(Y_1^n|T_1^n,Z_2^n) + H(Y_2^n)- H(Y_2^n|X_2^n,Z_2^n) + n\epsilon_n\nonumber\\
& = &  H(Y_1^n|T_1^n,Z_2^n) + H(Y_2^n) + n\epsilon_n\nonumber\\
& \leq & \sum_{i=1}^n H(Y_{1i}|T_{1i},W_i) + \sum_{i=1}^n H(Y_{2i}) + n\epsilon_n.\IEEEeqnarraynumspace\label{eq:conv5}
\end{IEEEeqnarray}

In \eqref{eq:conv1}-\eqref{eq:conv5}, we can introduce a time-sharing random variable $Q$ uniformly distributed on $[1:n]$. Defining $W$ to be $(W_Q,Q)$ and $(X_1,X_2,Y_1,Y_2)$ to be $(X_{1Q},X_{2Q},Y_{1Q},Y_{2Q})$, we get the required bounds on the rates.

We also require the Markov relationship $X_1-W-X_2$ to be satisfied. Since $W_i$ is chosen to be $Z_2^{i-1}$, we immediately have $X_{1i} - W_i - X_{2i}$ and hence $X_1-W-X_2$. Finally, the bound on the cardinality of $W$ can be established using the standard convex cover method.  \hfill\IEEEQED

\section{Extensions}\label{sec:extensions}

\subsection{State-Dependent Z-channel}

The result in Theorem~\ref{thm:cap} can be extended easily to Z-channels in which transmitter~1 also wishes to communicate a message $M_{21}$ at rate $R_{21}$ to receiver~2, as shown in Fig.~\ref{fig:model_state_det_zc}. 

\begin{figure}[!h]
\centering
\includegraphics[scale=1.5]{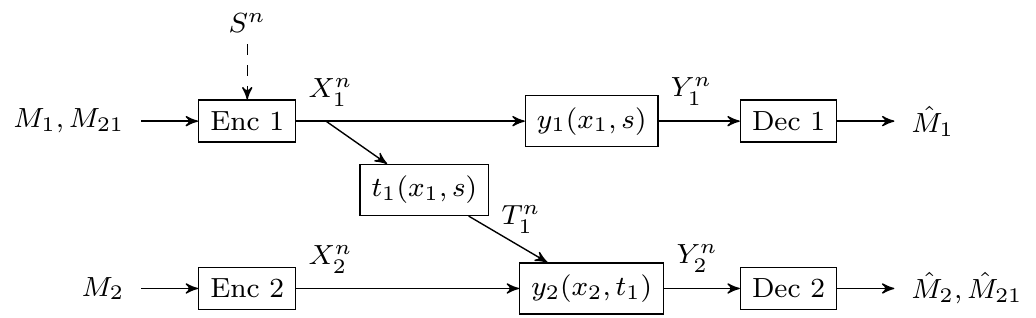}
\caption{The Injective Deterministic S-D Z-C}
\label{fig:model_state_det_zc}
\end{figure}  

\begin{thm}\label{thm:cap_z} The capacity region of the injective deterministic state-dependent Z-channel is the set of rate pairs $(R_{1},R_{21},R_{2})$ that satisfy
\begin{equation*}\label{eq:cap_z}
\begin{split}
R_{1} & \leq H(Y_1|S,Q),\\
R_{2} & \leq  H(Y_2|T_1,Q),\\
R_{21} & \leq H(T_1|S,Q),\\
R_{1} + R_{21} & \leq H(T_1,Y_1|S,Q),\\
R_{2} + R_{21} & \leq  H(Y_2|Q) - I(T_1;S|Q),\\
R_{1} + R_{2} + R_{21} & \leq H(Y_1|T_1,S,Q)+H(Y_2|Q)-I(T_1;S|Q),
\end{split}
\end{equation*}
for some pmf $p(q)p(x_1|s,q)p(x_2|q),$ where $|\mc{Q}|\leq 6$.
\end{thm}

The achievability scheme for this case is similar to the achievability scheme for the Z-IC described in Section~\ref{subsec:proof_thm_cap} except that now the $v^n$ sequences at encoder~1 are also binned, and this bin-index corresponds to the message $M_{21}$, that is to be communicated from transmitter~1 to receiver~2. The converse can be established by following similar steps as the converse for the Z-IC.

\begin{remark}
Since broadcast channel and multiple-access channel are special cases of the Z-channel, Theorem~\ref{thm:cap_z} also provides the capacity region for the deterministic broadcast channel and the injective deterministic multiple-access channel.
\end{remark}

\subsection{State-dependent injective deterministic Z-IC with unidirectional partial cribbing}

\begin{figure}[!ht]
\centering
\includegraphics[scale=1.5]{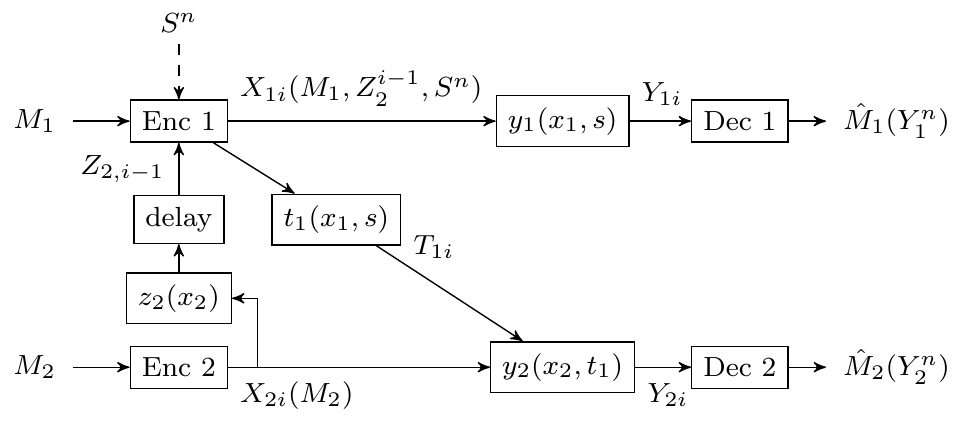}
\caption{State-Dependent Injective Deterministic Z-Interference Channel with Unidirectional Partial Cribbing}
\label{fig:model_state_crib}
\end{figure}

To illustrate further the advantages of the multicoding scheme, we consider a model that combines the state-dependent Z-IC and the Z-IC with unidirectional partial cribbing, as depicted in Fig.~\ref{fig:model_state_crib}. We can combine the achievability schemes for the two component setups from Sections~\ref{subsec:proof_thm_cap} and \ref{subsec:proof_crib} in a straightforward manner to get an achievability scheme for this setup. It turns out that this is capacity-achieving, resulting in the following theorem. 
\begin{thm} The capacity region of the channel in Fig.~\ref{fig:model_state_crib} is given by the convex closure of rate pairs $(R_1,R_2)$ satisfying
\begin{IEEEeqnarray*}{rCl}
R_1 & < & H(Y_1|W,S)\\
R_2 & < & H(Y_2,Z_2|T_1,W)\\
R_2 & < & \min(H(Y_2),H(Y_2,Z_2|W)) - I(T_1;S|W)\\
R_1 + R_2 & < & H(Y_1|W,T_1,S) + \min(H(Y_2),H(Y_2,Z_2|W)) - I(T_1;S|W)
\end{IEEEeqnarray*}
for pmf of the form $p(w)p(x_2|w)p(x_1|w,s)$.
\end{thm}

The proof of the converse combines ideas from the converse proofs we have presented for the state-only case and the cribbing-only case.

\section{Discussion}\label{sec:conclude}
Note that there is one difference between the multicoding-based achievability scheme for the canonical interference channel and the achievability schemes for the settings in Sections~\ref{sec:state} and \ref{sec:cribbing}. For the former, the codebook associated with the auxiliary random variable $U$ is used at both receivers during decoding, whereas for the cribbing setup, the codebook associated with auxiliary random variable $U_d$ is only used at the desired receiver, while that associated with the auxiliary random variable $U_c$ is only used at the undesired receiver. One way to understand this dichotomy is to observe similarities with the inner bound for broadcast channel which combines Marton coding and superposition coding, given in \cite[Proposition 8.1]{Gam12} which involves three auxiliary random variables $U_0,U_1,U_2$. Here, the random variable $U_0$ is used at both receivers during decoding, while $U_1$ and $U_2$ are used only at the respective receivers. Now for the deterministic cribbing setup that we have considered, we can think that in the optimal scheme, $U_0$ can be set to be the empty random variable $\phi$, and $U_1$ and $U_2$ correspond to $U_d$ and $U_c$ respectively, i.e., there is no superposition coding, only Marton coding is used (with the distinction from usual Marton coding that the set of $U_c^n$ sequences is not binned). The situation is similar for deterministic and semideterministic broadcast channels where $U_0=\phi$ is optimal too. On the other hand, the Han-Kobayashi scheme employing superposition coding can be thought of as setting $U_2$ to $\phi$, and $U_0$ and $U_1$ correspond to $U$ and $X$ respectively (no Marton coding, only superposition coding). The key observation is that in the Han-Kobayashi scheme, it is not necessary to think of $U$ as encoding a part of the message explicitly, which can be exploited to view the superposition coding instead in a manner resembling Marton coding, as we have shown in Section~\ref{sec:outline} (again with the distinction from usual Marton coding that the set of sequences corresponding to the auxiliary random variable is not binned). This clarifies the dichotomy mentioned at the beginning of this paragraph. Alternatively, we can understand both ways of decoding in a unified manner, as in \cite{Ban12}.

\section*{Acknowledgment}
We gratefully acknowledge discussions with Young-Han Kim, Chandra Nair, Shlomo Shamai and Tsachy Weissman. 

\appendix[Multivariate Packing Lemma]\label{app:multipack}
\begin{lem}
Consider the following four assumptions.
\begin{itemize}
\item[(A)] Let $\{U,X_1,X_2,\dots,X_K,Y\}$ be random variables that have some joint distribution $p_{U,X_1,X_2,\dots ,X_K,Y}$. 
\item[(B)] Let $(\tilde{U}^n,\tilde{Y}^n)\sim p(\tilde{u}^n,\tilde{y}^n)$ be a pair of arbitrarily distributed sequences. 
\item[(C)] For each $j\in [1:K]$, let $\{X_j^n(m_j),\; m_j\in\mc{A}_j\}$, where $|\mc{A}_j|\leq 2^{nR_j}$, be random sequences each distributed according to $\prod_{i=1}^n p_{X_j|U}(x_{ji}|\tilde{u}_i)$. 
\item[(D)] For each $j\in [1:K]$ and each $m_j$, assume that $X_j^n(m_j)$ is pairwise conditionally independent of $\left(\dots ,X_{j-1}^n(m_{j-1}),X_{j+1}^n(m_{j+1}),\dots ,\tilde{Y}^n\right)$ given $\tilde{U}^n$ for all $(\dots ,m_{j-1},m_{j+1}, \dots)$, but arbitrarily dependent on other $X_j^n(\cdot)$ sequences.
\end{itemize} 
Then there exists $\delta(\epsilon)$ that tends to zero as $\epsilon\rightarrow 0$ such that
\begin{IEEEeqnarray*}{LCl}
\msf{P}((\tilde{U}^n,X_1^n(m_1),X_2^n(m_2),\dots, X_K^n(m_K),\tilde{Y}^n)\in\mc{T}^{(n)}_{\epsilon}\\
\quad\quad\quad \text{ for some } m_1\in\mc{A}_1,m_2\in\mc{A}_2,\dots,m_K\in\mc{A}_K)
\end{IEEEeqnarray*}
tends to $0$ as $n\rightarrow\infty$ if \begin{equation}\label{eq:mvpack}\sum_{j=1}^K R_j < \sum_{j=1}^K H(X_j|U)- H(X_1,X_2,\dots ,X_K|U,Y) - \delta(\epsilon).\end{equation}
\end{lem}

\begin{proof}
The proof follows on similar lines as that of the packing lemma in \cite[Ch. 3]{Gam12}. 

Consider a fixed tuple $(\tilde{m}_1,\tilde{m}_2,\dots,\tilde{m}_K)$. The chain of inequalities resulting in \eqref{eq:multipack_proof} bounds the probability of $\left(\tilde{U}^n,\{X_j^n(\tilde{m}_j)\}_{j=1}^K,\tilde{Y}^n\right)$ being jointly typical, where $(a)$, $(b)$ and $(c)$ follow from assumptions (C) and (D). Then we can apply the union bound over all possible tuples $(\tilde{m}_1,\tilde{m}_2,\dots,\tilde{m}_K)$ to get the condition \eqref{eq:mvpack}.

\begin{figure*}[!h]
\normalsize
\hrulefill
\begin{IEEEeqnarray}{LCl}
\msf{P}\left(\left(\tilde{U}^n,\{X_j^n(\tilde{m}_j)\}_{j=1}^K,\tilde{Y}^n\right)\in\mc{T}^{(n)}_{\epsilon}\right) \nonumber\\
 =\sum_{(\tilde{u}^n,\tilde{y}^n)\in\mc{T}^{(n)}_{\epsilon}}p(\tilde{u}^n,\tilde{y}^n)\msf{P}\left(\left(\tilde{U}^n,X_1^n(m_1),X_2^n(m_2),\dots, X_K^n(m_K),\tilde{Y}^n\right)\in\mc{T}^{(n)}_{\epsilon}\Big |\tilde{U}^n=\tilde{u}^n,\tilde{Y}^n=\tilde{y}^n \right)\nonumber\\
 \stackrel{(a)}{=}\sum_{(\tilde{u}^n,\tilde{y}^n)\in\mc{T}^{(n)}_{\epsilon}}p(\tilde{u}^n,\tilde{y}^n)\msf{P}\left(\left(\tilde{u}^n,X_1^n(m_1),X_2^n(m_2),\dots, X_K^n(m_K),\tilde{y}^n\right)\in\mc{T}^{(n)}_{\epsilon}\Big |\tilde{U}^n=\tilde{u}^n \right)\nonumber\\
   \stackrel{(b)}{=} \sum_{(\tilde{u}^n,\tilde{y}^n)\in\mc{T}^{(n)}_{\epsilon}}p(\tilde{u}^n,\tilde{y}^n) \sum_{T^{(n)}_{\epsilon}\left(X_1,X_2,\dots ,X_K|\tilde{u}^n,\tilde{y}^n\right)} p(x_1^n,\dots , x_K^n|\tilde{u}^n)\nonumber\\
\stackrel{(c)}{=} \sum_{(\tilde{u}^n,\tilde{y}^n)\in\mc{T}^{(n)}_{\epsilon}}p(\tilde{u}^n,\tilde{y}^n) \sum_{T^{(n)}_{\epsilon}\left(X_1,X_2,\dots ,X_K|\tilde{u}^n,\tilde{y}^n\right)} \prod_{j=1}^K p(x_j^n|\tilde{u}^n)\nonumber\\
 \leq \sum_{(\tilde{u}^n,\tilde{y}^n)\in\mc{T}^{(n)}_{\epsilon}}p(\tilde{u}^n,\tilde{y}^n) \Big |T^{(n)}_{\epsilon}\left(X_1,X_2,\dots ,X_K|\tilde{u}^n,\tilde{y}^n\right)\Big | \prod_{j=1}^K 2^{-n(H(X_j|U)-\delta(\epsilon))}\nonumber\\
 \leq \sum_{(\tilde{u}^n,\tilde{y}^n)\in\mc{T}^{(n)}_{\epsilon}}p(\tilde{u}^n,\tilde{y}^n) 2^{n(H(X_1,\dots, X_K|U,Y)+\delta(\epsilon))} \prod_{j=1}^K 2^{-n(H(X_j|U)-\delta(\epsilon))}\nonumber\\
 = 2^{n(H(X_1,\dots, X_K|U,Y)+\delta(\epsilon))} \prod_{j=1}^K 2^{-n(H(X_j|U)-\delta(\epsilon))}\label{eq:multipack_proof}.
\end{IEEEeqnarray} 
\hrulefill
\vspace{-2.5mm}
\end{figure*}
\end{proof}

\subsection*{Example 1}
For the case $U=\phi$ and random variables $X_1,X_2,Y$, the condition \eqref{eq:mvpack} can be expressed as follows:
\begin{IEEEeqnarray}{rCl}
R_1+R_2 & < & H(X_1) + H(X_2) - H(X_1,X_2|Y)-\delta(\epsilon)\nonumber\\
& = & I(X_1;X_2) + I(X_1,X_2;Y) -\delta(\epsilon).\label{eq:multipack_2}
\end{IEEEeqnarray}
\subsection*{Example 2}
For the case $U=\phi$ and random variables $X_1,X_2,X_3,Y$, the condition \eqref{eq:mvpack} can be expressed as follows:
\begin{IEEEeqnarray}{rCl}
R_1+R_2+R_3 & < & H(X_1)+H(X_2)+H(X_3) \nonumber\\
&&\quad\quad -\> H(X_1,X_2,X_3|Y)- \delta(\epsilon)\nonumber\\
& = & I(X_1;X_2) + I(X_3;Y) \nonumber\\
&&\quad\quad +\> I(X_1,X_2;X_3,Y) - \delta(\epsilon).\IEEEeqnarraynumspace \label{eq:multipack_3}
\end{IEEEeqnarray}

\bibliographystyle{IEEEtran}
\bibliography{IEEEfull,interference}

\end{document}